\documentclass[a4paper,11pt]{article}

\usepackage[affil-it]{authblk}
\usepackage{setspace}
\usepackage{fullpage}
\usepackage{epsfig}
\usepackage{latexsym}
\usepackage{amsfonts}
\usepackage{amsmath}
\usepackage{amsthm}
\usepackage{amssymb}
\usepackage{amsbsy}
\usepackage{multirow}
\usepackage{slashed}
\usepackage{color}
\usepackage{mathrsfs}
\usepackage{subfigure}
\usepackage[font={footnotesize,sl},bf]{caption}
\usepackage{wrapfig}
\usepackage{pdflscape}

\usepackage[stable]{footmisc}
\usepackage{enumerate}

\def\be{\begin{equation}}
\def\ee{\end{equation}}
\def\bea{\begin{eqnarray}}
\def\eea{\end{eqnarray}}

\newcommand{\nn}{\nonumber \\}

\newtheorem{theorem}{Main Result}
\newtheorem{mathproposition}{Lemma}[section]

\newcommand{\YES}{$\bigcirc$}
\newcommand{\NO}{}
\newcommand{\INF}{$\bigtriangleup$}
\newcommand{\MAYBE}{$\left(\bigcirc\right)$}

\newcommand{\la}[1]{\mathbf{#1}}

\def\half{{\frac{1}{2}}}

\def\p{\partial}
\def\pb{\overline{\partial}}
\def\unit{{1\kern-.65ex {\rm l}}}
\def\1{{1\kern-.65ex {\rm l}}}

\def\ad{\mathop{\mathrm{ad}}\nolimits}
\def\adj{\mathop{\mathrm{adj}}\nolimits}

\def\dim{\mathop{\mathrm{dim}}\nolimits}
\def\diag{\mathop{\mathrm{diag}}\nolimits}
\def\End{\mathop{\mathrm{End}}\nolimits}

\def\rank{\mathop{\mathrm{rank}}\nolimits}

\def\Tr{\mathop{\mathrm{Tr}}\nolimits}
\def\tr{\mathop{\mathrm{tr}}\nolimits}
\def\inv{\mathop{\mathrm{inv}}\nolimits}

\def\zb{{\overline{z}}}
\def\Db{{\overline{D}}}
\def\Qb{{\overline{Q}}}
\def\Pb{{\overline{P}}}

\def\zetab{{\overline{\zeta}}}

\def\cM{{\cal M}}
\def\cN{{\cal N}}

\def\cP{{\cal P}}

\def\bbC{{\mathbb{C}}}

\def\bbR{{\mathbb{R}}}

\def\bbZ{{\mathbb{Z}}}

\begin{document}
 \title{Classifying Supersymmetric Solutions in 3D Maximal Supergravity}
 
 \author[1]{Jan de Boer%
  \thanks{Electronic address: \texttt{j.deboer@uva.nl}}}
  
 \author[1]{Daniel R.\ Mayerson%
  \thanks{Electronic address: \texttt{d.r.mayerson@uva.nl}}}
  
  \author[2,3]{Masaki Shigemori%
  \thanks{Electronic address: \texttt{shige@yukawa.kyoto-u.ac.jp}}}
  
\affil[1]{Institute for Theoretical Physics, University of Amsterdam, Science Park 904, Postbus 94485, 1090 GL Amsterdam, The Netherlands}

\affil[2]{Yukawa Institute for Theoretical Physics, Kyoto University,
Kitashirakawa Oiwakecho, Sakyo-ku, Kyoto 606-8502 Japan}
\affil[3]{Hakubi Center, Kyoto University,
Yoshida-Ushinomiyacho, Sakyo-ku, Kyoto 606-8501, Japan}

\date{}






\setcounter{tocdepth}{2}

\maketitle
\abstract{
String theory contains various extended objects. Among those, objects of
codimension two (such as the D7-brane) are particularly interesting.
Codimension two objects carry non-Abelian charges which are elements of
a discrete U-duality group and they may not admit a simple space-time
description, in which case they are known as exotic branes.  A complete
classification of consistent codimension two objects in string theory is
missing, even if we demand that they preserve some supersymmetry.
As a step toward such a classification, we study the supersymmetric
solutions of 3D maximal supergravity, which can be regarded as
approximate description of the geometry near codimension two objects.
We present a complete classification of the types of supersymmetric
solutions that exist in this theory.  We found that this problem reduces
to that of classifying nilpotent orbits associated with the U-duality
group, for which various mathematical results are known.  We show that
the only allowed supersymmetric configurations are 1/2, 1/4, 1/8, and
1/16 BPS, and determine the nilpotent orbits that they correspond to.
One example of 1/16 BPS configurations is a generalization of the MSW
system, where momentum runs along the intersection of seven
M5-branes.
On the other hand, it turns out exceedingly difficult to translate this
classification into a simple criterion for supersymmetry in terms of the
non-Abelian (monodromy) charges of the objects.
For example, it can happen that a supersymmetric solution exists locally
but cannot be extended all the way to the location of the object.  To
illustrate the various issues that arise in constructing supersymmetric
solutions, we present a number of explicit examples.

Key words: exotic branes, T-folds, U-folds, non-geometric branes, classification, supersymmetric solutions, supergravity, 3D

PACS numbers: 04.20.-q, 04.20.-Jb, 04.65.+e, 11.25.-w, 11.25.Uv}
\newpage
\tableofcontents
\newpage

\section{Introduction}

Supergravity is well-known to be able to capture non-perturbative
physics of string theory which is often difficult to see in the full
theory.
For example, the existence of $p$-branes charged under Ramond-Ramond
gauge fields is predicted by the non-perturbative dualities of string
theory, but they were first found in supergravity as solitonic solutions
\cite{Horowitz:1991cd} before they were identified with the D-branes
\cite{Polchinski:1995mt} in perturbative string theory. 

By now, a vast number of supersymmetric solutions of supergravity have been found
in various dimensions, and in several cases a complete classification has been obtained.
However, so far, only a few solutions of $d=3$ maximally supersymmetric supergravity have been
constructed, which is unfortunate, since this is a very interesting case for reasons we now 
explain.

The duality of string theory predicts not only the standard branes such
as D-branes but also exotic branes \cite{Elitzur:1997zn, Blau:1997du,
Hull:1997kb, Obers:1997kk, Obers:1998fb} that have been
studied far less.  Exotic branes are codimension-two objects whose
higher-dimensional origin cannot be explained in terms of standard
branes (namely, D-branes, M-branes, F1-string, NS5-branes, gravitational
waves, and KK monopoles). Their characteristic feature is that they have a
non-trivial monodromy of $U$-duality around them.  Namely, as one goes
around an exotic brane, the spacetime fields do not come back to the
original value but only to the $U$-dual version.
Perhaps the most famililar example of such codimension-two branes is the
$(p,q)$ 7-brane in type IIB string theory which is well-known in the
context of F-theory \cite{Greene:1989ya, Vafa:1996xn}. Type IIB string
theory has $SL(2,\bbZ)$ as the $U$-duality group under which the
axion-dilaton field $\tau=C_0+ie^{-\Phi}$ transforms as
\begin{align}
 \tau\to \tau={a\tau+b\over c\tau+d},\qquad
 \begin{pmatrix} a&b\\c&d \end{pmatrix}\in SL(2,\bbZ).\label{SL(2,Z)_on_tau}
\end{align}
Two different values of $\tau$ related by \eqref{SL(2,Z)_on_tau} are to
be physically identified in string theory.  The $(p,q)$ 7-brane is a
codimension-two brane around which the field $\tau$ undergoes the
particular monodromy\footnote{The status of the $SL(2,\bbZ)$ 7-branes
with more general monodromies \eqref{SL(2,Z)_on_tau} in string theory is
unclear although they certainly exist at the level of supergravity
\cite{Bergshoeff:2006jj}.}
\begin{align}
 \begin{pmatrix} a&b\\c&d \end{pmatrix}
 =
 \begin{pmatrix}
  1+pq &p^2 \\
 -q^2 & 1-pq\end{pmatrix}
 \in SL(2,\bbZ).
\end{align}

One intriguing property of exotic branes is that they are generically
\emph{non-geometric} in the following sense.  Recall that $U$-duality
mixes internal components of higher-dimensional fields, including the
metric.
Therefore, the fact that exotic branes have non-trivial $U$-duality
monodromy around them implies that the metric is not single-valued in
the presence of exotic branes.  In this sense, exotic branes are
generically non-geometric \cite{deBoer:2010ud, deBoer:2012ma}. (There
are codimension-two branes that are geometric, such as the $(p,q)$
7-brane above though.)
Furthermore, these exotic branes are highly non-perturbative in the
sense that they typically have tension proportional to $g_s^{-3}$ or
$g_s^{-4}$.

\begin{table}[htbp]
\begin{center}
  \begin{tabular}{|@{\,}c@{\,}||@{\,}c@{\,}|@{\,}c@{\,}|@{\,}c@{\,}|@{\,}c@{\,}|@{\,}c@{\,}|@{\,}c@{\,}|}
  \hline
  $d$& $G(\bbR)$ & $G(\bbZ)$ & $K$  & $\dim G$ & $\rank G$ & $\dim K$ \\  \hline  \hline
  10A& $\bbR_+$                     & $\bf 1$                      & $\bf 1$             & 1   &1    & 0   \\ \hline
  10B& $SL(2,\bbR)$                 & $SL(2,\bbZ)$                 & $SO(2)$             & 3   &1    & 1   \\ \hline
  9  & $SL(2,\bbR)\times \bbR_+$    & $SL(2,\bbR)\times \bbZ_2$    & $SO(2)$             & 4   &2    & 1   \\ \hline
  8  & $SL(3,\bbR)\times SL(2,\bbR)$& $SL(3,\bbZ)\times SL(2,\bbZ)$& $SO(3)\times SO(2)$ &$8+3$&$2+1$&$3+1$\\ \hline
  7  & $SL(5,\bbR)$                 & $SL(5,\bbZ)$                 & $SO(5)$             & 24  &4    & 10  \\ \hline
  6  & $SO(5,5,\bbR)$               & $SO(5,5,\bbZ)$               & $SO(5)\times SO(5)$ & 45  &5    & 20  \\ \hline
  5  & $E_{6(6)}$                   & $E_{6(6)}(\bbZ)$             & $USp(8)$            & 78  &6    & 36  \\ \hline
  4  & $E_{7(7)}$                   & $E_{7(7)}(\bbZ)$             & $SU(8)$             &133  &7    & 63  \\ \hline
  3  & $E_{8(8)}$                   & $E_{8(8)}(\bbZ)$             & $SO(16)$            &248  &8    & 120  \\ \hline
 \end{tabular}
\end{center}
 \vspace*{-3ex}
\caption{\sl The $U$-duality groups in maximally supersymmetric
supergravity and string theory in various ($d$) dimensions. $G(\bbR)$ is
the classical $U$-duality group while $G(\bbZ)$ is the quantum one
\cite{Hull:1994ys}. $K$ is the maximal compact subgroup of
$G(\bbR)$.\label{table:CJgroups}}
\end{table}

Because of these facts, exotic branes are difficult to analyze in
perturbative string theory.\footnote{Note however some recent work on
sigma-model description of the $5^2_2$ exotic brane which has tension
proportional to $g_s^{-2}$ and presumably allow for perturbative description
\cite{Kimura:2013fda, Kimura:2013zva, Kimura:2013khz}.} However, in
supergravity, they are represented simply by solutions with non-trivial
monodromies for scalars around them, just as in \eqref{SL(2,Z)_on_tau}.
Therefore, it is very interesting to ask what is the most general
solution possible in supergravity with non-trivial $U$-duality
monodromies.  In this paper, we attempt such a classification of
codimension-two branes.  The expectation is that the solutions we find
correspond to non-perturbative objects in the full string theory.

Note that the problem of classifying codimension-two branes is more
non-trivial than that of classifying higher codimension branes.  This is
because the charge of higher codimension branes is measured by usual
gauge field flux and lives in a linear lattice $\bbZ^n$ where $n$ is the
number of gauge fields, while the charge of codimension-two branes is
measured by its monodromy which lives in the $U$-duality group
$G(\bbZ)$.  The discrete non-Abelian group $G(\bbZ)$ has more intricate
structure than $\bbZ^n$ and the possible configurations of
codimension-two branes are more complicated than those of higher
codimension branes. Note that in this paper we are considering the (classical) supergravity moduli space, with the classical $U$-duality group $G(\bbR)$.

If we compactify type II string theory on $T^{10-d}$ or M-theory on
$T^{11-d}$ down to $d$ dimensions, then the $U$-duality group $G(\bbZ)$
becomes larger (see Table \ref{table:CJgroups}), the number of scalars
that span the group's representation increases, and, consequently,
the spectrum of codimension-two exotic branes becomes much richer. As can be seen in Table \ref{table:CJgroups}, the $U$-duality group in
lower dimensions contains the $U$-duality group in higher dimensions.
As a result, codimension-two objects in higher dimensional theory always
exist in lower dimensional theory too.  Therefore, for the purpose of
classifying codimension-two branes, we can just focus on the lowest
possible dimension, namely $d=3$.  For $d=3$, codimension-two objects
are point particles and therefore our goal is to classify configurations
of particle-like objects in 2+1 dimensions. Note that a point particle in $d=3$ is going to destroy flat
asymptotics. Therefore, we should regularize geometry at long distance
by closing it to $S^2$ as in F-theory \cite{Greene:1989ya, Vafa:1996xn}
or regard our solutions as near-source approximation of some higher
codimension configurations such as supertubes \cite{deBoer:2010ud,
deBoer:2012ma} which have flat asymptotics.

To achieve the goal of classifying all possible point particles in $d=3$, a first step is to classify all possible \emph{supersymmetric} point particles. A classification of all half-supersymmetric point particles in $d=3$ was already achieved in \cite{Bergshoeff:2013sxa,Bergshoeff:2012ex,Bergshoeff:2011se}, but a full classification of all possible point particles preserving any supersymmetry is still lacking. In this paper, we try to remedy this somewhat by finding a full classification of supersymmetric solutions in $d=3$ maximally supersymmetric supergravity, which should be the low-energy limit of the compactification of type II string theory on a $T^7$. This full classification entails the precise necessary and sufficient conditions for supersymmetry in this supergravity theory. While our classification is thus complete in this sense, unfortunately this does not specify the allowed monodromies/charges of supersymmetric point particles. This is because our classification of supersymmetry is given in terms of a quantity $P_z$ which is not immediately related to the monodromy/charge of point particles.\footnote{This can be related to the discussion in \cite{Bergshoeff:2013sxa}, in which is shown using group theory arguments that the BPS condition is degenerate for branes with codimension two or less, i.e.\ that there are multiple branes that preserve exactly the same supersymmetries.} Moreover, even if a local supersymmetric 
solution exists, there may exist global obstructions towards extending the solution all the way to the location of the point
particle.

In lieu of a full classification of the allowed monodromies of supersymmetric point particles, we will give a plethora of (not only supersymmetric) examples of point particle solutions with various monodromies, which will also further illustrate the difficulty of achieving a full classification of allowed monodromies of supersymmetric solutions.

A solution in $d=3$ dimensions can be uplifted to a configuration in 10
or 11 dimensions involving various branes.  Note, however, such an
uplift is not unique but is determined only up to an overall $U$-duality
transformation.  Therefore, an interesting question to ask is if there
are 3d configurations whose uplift involve non-geometric exotic branes
in all possible $U$-duality frames. However, since, as we will show, the
relation between the brane monodromy and supersymmetry is unclear, we
have no straightforward way to answer to this question. The only thing
we can say is that we have succeeded in finding a ``standard''
(non-exotic) brane representative for each of the supersymmetry
preserving orbits (see especially Table \ref{tab:orbits}). However, even if all supersymmetric 
branes were to admit a standard representative, this would still not imply that combinations
of different supersymmetric branes would simultaneously admit a representation in terms of standard branes.
Unfortunately, the study of multi-centered solutions appears to be at least one order of magnitude
more complicated than the construction of single point particles, and we will not attempt to do so
in this paper.

As already mentioned, we \emph{do} manage to find a complete classification in terms of explicit necessary and sufficient conditions for supersymmetry in 3D maximal supergravity. Our initial analysis follows the strategy pioneered in \cite{Tod:1983pm,Tod:1995jf}; we assume the existence of a Killing spinor and construct a Killing vector $V$ from it. This divides the supersymmetric solutions into a ``timelike'' or ``null'' class, depending on if $V$ is timelike or null. The situation is reminiscent of classifications of solutions in minimal supergravity in $d=5,6$ and $\mathcal{N}=2$ supergravity in $d=4$, see e.g.\ \cite{Gauntlett:2002nw,Gauntlett:2003fk,Gutowski:2003rg,Bellorin:2006yr}. However, the structure of maximal supergravity is more complicated than in these situations with considerably less supersymmetry, and we must resort to a detailed analysis of a quantity $P_z$ (which takes values in the algebra $\la{e}_8$) to further determine conditions of supersymmetry preservation. It turns out that $P_z$ must be a nilpotent algebra element to preserve supersymmetry, and moreover there are only a few nilpotent orbits that preserve a given amount of supersymmetry (all listed in table \ref{tab:orbits}). This, in turn, is reminiscent of \cite{Bossard:2009at,Bossard:2009we}, where the classification of supersymmetric stationary solutions in 4D is studied via a timelike reduction from to 3D, giving a \emph{Euclidean} supergravity theory in which nilpotent orbits also play a crucial role in the classification of supersymmetry preservation. We do note that the analysis presented in \cite{Bossard:2009at,Bossard:2009we} is quite different, as e.g.\ the $U$-duality structure in this Euclidean theory admits a non-compact version of $\la{so}(16)$ (which should be compared to the compact $\la{so}(16)$ of table \ref{table:CJgroups} for 3D)\@.
We also note that a partial classification of supersymmetric solutions in
half-maximal gauged supergravity has been done in \cite{Deger:2010rb},
where $G=SO(8,n)$, $H=SO(8)\times SO(n)$.  Their analysis is
complementary to ours in that they did a more detailed analysis of the spinor
bilinears while we focused on Lie-algebraic structures.  It
would be interesting to combine the techniques of that paper with ours
to clarify more structure of supersymmetric solutions.

%

%
%

The organization of the rest of this paper is as follows. In section
\ref{sec:preliminaries}, we first review some important mathematical and
physical structures we need for our story. We review facts about general
scalar cosets in $d$ dimensions before specializing to the case at hand
in $3d$ with coset $E_{8(8)}/SO(16)$; finally, we also review concepts
in Lie algebras that are crucial for our story---most notably about the
(nilpotent) orbit structure in real and complex Lie algebras. Section
\ref{sec:allresults} contains the classification of supersymmetric
solutions in 3D maximal supergravity; all results are formulated and
summarized in section \ref{sec:statementresults}, which can be seen as
the main result in this paper; the proofs of the results are given in
section \ref{sec:derivationresults}. In sections \ref{sec:simplebranes}
and \ref{sec:morebranes}, we search for explicit single-center brane
solutions with various scalar monodromies. Section
\ref{sec:simplebranes} deals with setting up an ansatz for such
single-center solutions, and finding recipes to construct branes with
any semisimple or nilpotent monodromy; it is also explained how the
brane representations of table \ref{tab:orbits} are obtained. In section
\ref{sec:morebranes}, we try to find more complicated single-center
brane solutions that live in $\la{sl}(3)$. The results of sections
\ref{sec:simplebranes} and \ref{sec:morebranes} are summarized in
section \ref{sec:potformtaxonomy} and especially table
\ref{tab:taxonomy}. Finally, we summarize and discuss our results in
section \ref{sec:conclusion}.

For a first reading of this paper, we suggest first browsing the
preliminaries in section \ref{sec:preliminaries}, if necessary. A
self-contained overview of our results can be obtained by reading the
main results of the supersymmetry classification in section
\ref{sec:statementresults}, and the summary of the single-centered brane
solutions (of sections \ref{sec:simplebranes} and \ref{sec:morebranes})
in section \ref{sec:potformtaxonomy} and especially table
\ref{tab:taxonomy}.

\section{Preliminaries}
\label{sec:preliminaries}

Below, we first review maximally supersymmetric supergravity (maximal
supergravity), focusing on the structure of the scalar sector.  After
discussing general $d$ dimensions, we will turn to the $d=3$ theory
which we are interested in. Then, we give a brief summary of some facts
about Lie algebras of relevance to us later.

\subsection{Scalar cosets and maximal supergravity}

If we compactify 10-dimensional type IIA/B supergravity on $T^{10-d}$ or
11-dimensional supergravity on $T^{11-d}$ down to $d$ non-compact
dimensions, we obtain maximal (ungauged) supergravity with the duality symmetry
groups $G(\bbR)$ as summarized in Table \ref{table:CJgroups}.  These
theories have scalar fields parametrizing the coset space
$\cM_c=G(\bbR)/K(\bbR)$, where $K(\bbR)$ is the maximal compact subgroup
of $G(\bbR)$.  The isometry group of $\cM_c$ is $G(\bbR)$.  In string
theory, the duality group reduces to the discrete subgroup $G(\bbZ)$
and, at the same time, the scalar moduli space becomes $\cM_q=G(\bbZ)
\backslash G(\bbR)/K(\bbR)$.  Namely, points in the classical moduli
space $\cM_c$ related by an element of $G(\bbZ)$ are identified in the
quantum moduli space $\cM_q$. We will not see these quantum effects in
our analysis, as we are considering classical solutions in supergravity,
i.e.\ working in the classical moduli space $\cM_c$.

Let us study the coset structure of the scalar sector of the theory in
more detail \cite{Marcus:1983hb, Pope:KKtheory}.  Let us denote the Lie
algebra of $G$ and $K$ by $\la{g}$ and $\la{k}$,
respectively. In $\la{g}$, we can define the Cartan involution
$\theta$ which reverses the sign of every non-compact generator while
leaves unchanged the sign of every compact generator.  Using $\theta$,
all the Cartan generators of the Lie algebra $\la{g}$ decompose as
(Cartan decomposition)
\begin{align}\label{eq:cartandecomp}
\la{g}=\la{k}\oplus \la{p},\qquad\qquad
 \la{p}\equiv  \la{g}\ominus \la{k},
\end{align}
where
\begin{align}
\theta(\la{k})=\la{k}\quad \text{(compact)},\qquad
\theta(\la{p})=-\la{p}\quad\text{(non-compact)}.
\end{align}
$\theta$ gives a $\mathbb{Z}_2$ grading because $[\la{k},\la{k}]\in \la{k}$, $[\la{k},\la{p}]\in \la{p}$, $[\la{p},\la{p}]\in \la{k}$;
we will call elements of $\la{k}$ even generators and ones of $\la{p}$ odd
generators.  The real Lie algebra $\la{g}$ that appears in maximal
supergravity is the maximally non-compact real form (also known as the
split real form). In this case, all Cartan generators and half of the
other generators are non-compact, while the other half are compact.
More precisely, let us denote the positive-root generators,
negative-root generators and Cartan generators in the Cartan-Weyl basis  by
$(E_{\alpha},E_{-\alpha},H^i)$, where $\alpha$
ranges over all the positive roots.  For $\la{g}$ maximally non-compact,
$\theta$ can be taken to act as
\begin{align}
 \theta~:~~ (E_{\alpha},E_{-\alpha},H^i)
 ~\to~
 (-E_{-\alpha},-E_{\alpha},-H^i).
\end{align}
Therefore, $\vec{H}$ and $E_{\alpha}+E_{-{\alpha}}$ are
non-compact (odd) while $E_{\alpha}-E_{-\alpha}$ are compact
(even).
We can take $\theta(x)=-x^T$ ($x\in \la{g}$) if $K$ is
an orthogonal group
while
 $\theta(x)=-x^\dagger$  if $K$ is
a unitary group.
Note that $\theta$ satisfies $\theta([x,y])=[\theta(x),\theta(y)]$.

The scalar fields of maximal supergravity parametrize the coset space
$\cM_c=G(\bbR)/K(\bbR)$.  They can be represented by a matrix $V \in
G(\bbR)$ if we assume that $g\in G(\bbR)$ and $k\in K(\bbR)$ act on $V $ as
follows:
\begin{align}
 V \to g V  k.
 \label{V->GVH}
\end{align}
Two matrices $V ,V '$ must be identified if they are related by
right-multiplication of $k\in K$.  To fix this ambiguity, we must
choose some specific gauge.  One convenient gauge choice is the ``unitary
gauge''
\cite{Marcus:1983hb}
in which
\begin{align}
 V =e^{\phi}, \qquad \phi\in \la{p}=\la{g}\ominus \la{k}.\label{unitary_gauge}
\end{align}
When we act on $V$ with a global $G$-duality transformation $g\in
G(\bbR)$ according to \eqref{V->GVH}, we must also act with a
compensating local $K$ transformation so that we remain in the gauge
\eqref{unitary_gauge}.  The local $K$ transformation in general
depends on the field $\phi$.

Another useful gauge uses the Iwasawa decomposition which says that we
can always write $V \in G$ as
\begin{align}
 V =n a k,
\label{eq:V=nak}
\end{align}
where $n$ is
generated by the positive roots of $\la{g}$,  $a$ is generated by the Cartan subalgebra of $\la{g}$, and $k\in K$.  By an
appropriate choice of the local $K$ transformation, we can always take
the ``Borel gauge'' in which the scalar moduli matrix takes the form
\begin{align}
 V =na.
\end{align}
One possible choice is to take
\begin{align}
 V =na=\Bigl[\prod_{\alpha>0}e^{A_{\alpha} E_{\alpha}}\Bigr]\, e^{\phi_i H^i },
\end{align}
where the product is over positive roots (for precise ordering of roots,
see \cite{Pope:KKtheory}).    The relation between the
lower dimensional scalars and the
internal components of higher dimensional fields
is easier to see in the Borel gauge \cite{Pope:KKtheory}.

An important quantity for constructing a covariant action is
 \begin{align}
  V ^{-1}\partial_{\mu} V =
 P_\mu+Q_\mu ,
\end{align}
where $P_\mu$ and $Q_\mu$ are the projection of $V^{-1}\partial_\mu V$ onto
$\la{p}$ and $\la{k}$, respectively; namely,
\begin{align}\label{eq:defPQ}
 P_\mu\equiv {1-\theta\over 2}V ^{-1}\partial_\mu V,\qquad
 Q_\mu\equiv {1+\theta\over 2}V ^{-1}\partial_\mu V.
\end{align}
These can be shown to satisfy Bianchi identities
\begin{align}
\begin{split}
 [D_\mu,D_\nu]&=-[P_\mu,P_\nu],\qquad D_\mu \equiv \partial_\mu + Q_\mu;\\
 D_\mu P_\nu &\equiv  \partial_\mu P_\nu +[Q_\mu,P_\nu]=D_\nu P_\mu.
\end{split}
\end{align}
Or equivalently, in form language,
\begin{align}
\begin{split}
  dQ+\half [Q,Q]+\half [P,P]&=0,\\
 dP+[Q,P]&=0,
\end{split}\label{bianchi_id_form}
\end{align}
where $P\equiv P_\mu dx^\mu, Q\equiv Q_\mu dx^\mu$, so that e.g.\ $[Q,P]=[Q_{\mu},P_{\nu}]dx^{\mu}\wedge dx^{\nu}$.

As explained above, a global $G$ transformation will induce a local $K$
transformation.  Under this local $K$ transformation, it can be shown
that $P_\mu$ transforms covariantly while $Q_\mu$ transforms as a $K$
gauge field.  This fact makes it possible to write down an invariant
action.  Because $\tr(P^\mu P_\mu)$ is invariant, the general form of
the metric and scalar part of the maximal supergravity action is
\begin{align}
 S_{\text{scalar}}
 ={1\over 4\kappa^2}\int d^dx \sqrt{-g}\left[R
 -g^{\mu\nu}\tr(P_\mu P_\nu)\right].\label{S_scalar_ito_P}
\end{align}
Our convention is that the signature of the metric is mostly plus.
Also, if we have a quantity that transforms under $K$ transformation,
i.e., $R$-symmetry, such as the gravitino, we can covariantize the
derivatives acting on it with respect to $K$ transformations using $Q$
as the gauge connection.
The total action has, in addition to the metric and scalars appearing in
$S_{\text{scalar}}$ \eqref{S_scalar_ito_P}, terms that involve other bosonic
form fields as well as fermions, all covariantized by the procedure just
outlined.
However, in this paper, we focus on the $d=3$ case where there are no
form fields but only scalars, and hence \eqref{S_scalar_ito_P} is the
full bosonic action.
Even for $d>3$, as long as one considers configurations with vanishing
form fields, \eqref{S_scalar_ito_P} is the relevant bosonic action.
In such situations, the equations of motion derived from this action are
\begin{align}
  R_{\mu\nu}-\tr(P_\mu P_\nu)&=0,\label{einstein_eq}\\
 \nabla_\mu P^\mu + [Q_\mu,P^\mu]&=0.\label{scalar_eom_P,Q}
\end{align}
where $\nabla_\mu$ is the covariant derivative with respect to the
Levi-Civita connection.

Rather than $P_\mu$ and $Q_\mu$ which depend on the gauge choice, it is
sometimes more useful to work with gauge independent quantities.  It is
clear that the quantity
\begin{align}
 M\equiv  V  V ^T.
\end{align}
is independent of the $K$ gauge transformation in
\eqref{V->GVH} if $K$ is the orthogonal group (if $K$ is unitary, use
$M=VV^\dagger$ instead).   Under the transformation
\eqref{V->GVH}, the matrix $M$ transforms as
\begin{align}
\label{eq:Mmong} M\to gMg^T.
\end{align}
In the unitary gauge \eqref{unitary_gauge}, $M$ can be written as
\begin{align}
\label{eq:Mdef}
 M=e^\phi (e^\phi)^T
 =e^\phi e^{-\theta(\phi)}
 =e^{2\phi}.
\end{align}
The action \eqref{S_scalar_ito_P} can be written in terms of $M$ as (using $2e^{\phi} P_{\mu} e^{-\phi} = \partial_{\mu}M M^{-1}$):
\begin{align}
 S_{scalar}
 ={1\over 4\kappa^2}\int d^dx \sqrt{-g}\left[R
 +{1\over 4}g^{\mu\nu}\tr(\partial_\mu M^{-1}\partial_\nu M)\right].
 \label{S_scalar_ito_M}
\end{align}
The equations of motion derived from this action are
\begin{align}
 R_{\mu\nu}-{1\over 4}
 \tr(\partial_\mu M\, M^{-1} \partial_\nu M\, M^{-1})&=0,\\
 \nabla_\mu (\nabla^\mu M\, M^{-1})&=0.\label{M_EOM}
\end{align}

The $M$ equation of motion \eqref{M_EOM} can be thought of as the
conservation law for the current 1-form $j$ as follows
\begin{align}
d*j=0,\qquad
 j\equiv \half dM\, M^{-1}= V P V^{-1}.\label{j_EOM}
\end{align}
From the definition of $j$, we can show 
that, if we move along a path parametrized by $\lambda$,
\begin{align}
 M(\lambda)=
 \cP e^{2\int_0^\lambda j} M(0),\label{M_transport}
\end{align}
where $\cP$ is path ordering with
respect to $\lambda$ and gives the monodromy\footnote{This monodromy is not the same as the usual scalar monodromy $g$ as in (\ref{eq:Mmong}), which is the usual definition of the scalar monodromy and is the one we will use in the rest of the paper.} of the matrix $M$.
From the definition of $j$, it follows that
$
dj+2j\wedge j=0.$
Therefore, $2j$ is a ``flat connection'' and the ``Wilson line''
$\cP e^{2\int_0^\lambda j}$ depends only
on the endpoints of the path.
Note that, using the relation $j M =M j^T$ which
immediately follows from the definition of $j$ and from the fact that
$M^T=M$, we can rewrite \eqref{M_transport} as
\begin{align}
 M(\lambda)=
 g(\lambda) M(0) g(\lambda)^T,\qquad
 g(\lambda)=\cP e^{\int_0^\gamma j}.\label{M_transport2}
\end{align}
Unlike $ \cP e^{2\int_\gamma j} $, the quantity $g(\lambda)$
does depend on the path $\gamma$, since $j$ is not a flat connection. 
Note that the relations \eqref{M_transport}, \eqref{M_transport2} follow
from definitions and are independent of the equation of motion
\eqref{j_EOM}.

\subsection{Maximal supergravity in three dimensions}\label{sec:maxsugra}

Thus far, we have been considering general $d$.  From now on, let us 
focus on the $d=3$ case which we are interested in.  Maximal
supergravity in $d=3$ was first constructed in \cite{Marcus:1983hb}.

For $d=3$, the duality group is $G=E_{8(8)}(\bbR)$ with the maximal
compact subgroup $K(\bbR)=SO(16)$.  The 248 generators of
$\la{g}=\la{e}_{8(8)}$ consist of the 120 compact
$\la{k}=\la{so}(16)$ generators $X^{IJ}$ ($I,J=1,\dots,16$,
$X^{IJ}=-X^{JI}$) and the 128 non-compact $\la{p}=\la{g}\ominus \la{k}$ generators $Y^A$
($A=1,\dots,128$) which transform in the Majorana-Weyl spinor
representation $\bf 128$ of $SO(16)$.
Correspondingly, the $P,Q$ fields introduced before can be expanded as
$P_\mu=P_\mu^A Y^A$, $Q_\mu=\half Q_\mu^{IJ}X^{IJ}$.  The scalar field
$\phi$ in \eqref{unitary_gauge} can be expanded as $\phi=\phi^A Y^A$.
For more details about the $\la{e}_{8(8)}$ Lie algebra, see appendix \ref{sec:appconstruction}.

The $d=3$ spacetime spinors can be taken to be two-component
Majorana spinors.  Since we have $\cN=16$ supersymmetry in maximal
supergravity, we have 16 gravitinos $\psi^I_\mu$ and 16 supersymmetry
transformation parameters $\epsilon^I$, where $I=1,\dots,16$.  In
addition, we have dilatinos $\chi^{\dot{A}}$ where $\dot{A}=1,\dots,128$
is the index for the other Majorana-Weyl spinor representation $\bf
128'$ of $SO(16)$.
We take the $d=3$ gamma matrices in
Minkowski spacetime to be
\begin{align}
 \gamma_{\hat{0}}&=\begin{pmatrix}0&1\\-1& 0\end{pmatrix},\qquad
 \gamma_{\hat{1}} =\begin{pmatrix}0&1\\ 1& 0\end{pmatrix},\qquad
 \gamma_{\hat{2}} =\begin{pmatrix}1&0\\ 0&-1\end{pmatrix}
\end{align}
where the hats mean flat indices.  In this basis, Majorana spinors are
really real.

The bosonic action \eqref{S_scalar_ito_P} must be supplemented with
fermionic terms for the total action to be supersymmetric.  We do not
need the form of the full action but let us note that, when the fermion
background vanishes, the supersymmetry transformations of the fermionic
fields are
\begin{align}
 \delta \psi^I_\mu &
 ={1\over \kappa}\left(\partial_\mu \epsilon^I
 +{1\over 4}\omega_{\mu}^{\hat{a}\hat{b}}\gamma_{\hat{a}\hat{b}}\epsilon^I
 +Q^{IJ}_\mu \epsilon^J
 \right),
 \label{susy_var_psi}
 \\
 \delta \chi^{\dot{A}}&={i\over 2\kappa}\gamma^\mu \epsilon^I \Gamma^I_{\dot{A}A}P^A_\mu.\label{susy_var_chi}
\end{align}
Here, $\omega_\mu^{\hat{a}\hat{b}}$ is the spin connection,
$\gamma_{\hat{a}\hat{b}}\equiv \half(\gamma_{\hat{a}} \gamma_{\hat{b}}-\gamma_{\hat{b}} \gamma_{\hat{a}})$, and
$\Gamma^I_{\dot{A}A}$ is the chiral block of the $SO(16)$ gamma matrices
defined in appendix \ref{app:appgammaso16} and can be taken to be real matrices.  The last term
in \eqref{susy_var_psi} is the $K$-covariantization mentioned above. For
the background to preserve supersymmetry, the above supersymmetry
variation must vanish.

In particular, let us consider the following configuration:
\begin{align}
 ds^2&=-dt^2+e^{U(z,\zb)}dz d\zb,\qquad
 P_t=Q_t=0,
 \label{static_config_a=1}
\end{align}
where $z=x^1+ix^2$, $\zb=x^1-ix^2$.  Later, we will see that requiring
supersymmetry leads to an ansatz of the form
\eqref{static_config_a=1}.\footnote{More precisely, this corresponds to
supersymmetric solutions in the timelike class.  There is also the null
class of solutions which is discussed in section \ref{sec:allresults}.} 
%
For the ansatz \eqref{static_config_a=1}, it is easy to see that the
field equations \eqref{einstein_eq}, \eqref{scalar_eom_P,Q} become
 \begin{align}
 \label{eq:EOMtimelike} \tr(P^2)= \tr(\Pb^2)&=0,\qquad \tr(P\Pb) = R_{z\zb} = -\partial\overline{\partial} U.\\
 \partial\Pb + [Q,\Pb] + \overline{\partial} P + [\Qb,P] &=0
 \end{align}
and the Bianchi identities \eqref{bianchi_id_form} are
\begin{align}
 \p\Pb-\pb P+[Q,\Pb]-[\Qb,P]=0,\qquad
 \p\Qb-\pb Q+[Q,\Qb]+[P,\Pb]=0,
\end{align}
where we used the shorthand notation
\begin{align}
 P\equiv P_z,\quad  \Pb\equiv P_\zb,\quad
 Q\equiv Q_z,\quad  \Qb\equiv Q_\zb,\quad
 \p\equiv \p_z,\quad \pb\equiv \p_\zb.
\end{align}
It will be important to note that $P=P_z$ and $Q=Q_z$ no longer
live in the real algebras
$\la{p}_\bbR = (\la{e}_{8(8)})_\bbR \ominus \la{so}(16)_\bbR$ and
$\la{k}_\bbR = \la{so}(16)_\bbR$, but in their complexified version
$\la{p}_\bbC = (\la{e}_{8(8)})_\bbC\ominus \la{so}(16)_\bbC$ and
$\la{k}_\bbC = \la{so}(16)_\bbC$,
respectively.

The condition for the field configuration to preserve supersymmetry is
that the supersymmetry variation \eqref{susy_var_psi},
\eqref{susy_var_chi} for fermions vanish.  Namely,
\begin{align}
 \zetab^I\Gamma^I_{\dot{A}A}P^A&=
 \zeta^I\Gamma^I_{\dot{A}A}\Pb^A=0,
 \label{delta_chi=0}
 \\
 D(e^{-U/4}\zeta)
 &=\Db(e^{U/4}\zeta)
 =D(e^{U/4}\zetab)
 =\Db(e^{-U/4}\zetab)=0,
 \label{delta_psi=0}
\end{align}
where we defined
\begin{align}
 \zeta^I\equiv \epsilon^I_1+i\epsilon^I_2,\qquad
 \zetab^I\equiv \epsilon^I_1-i\epsilon^I_2.
\end{align}
Note that the subscript of $\epsilon^I_{1,2}$ is the 3D spinor index.
Also, in \eqref{delta_psi=0}, $D,\Db$ are the $K$-covariant derivatives, $D
\xi^{I}=\p \xi^I+Q^{IJ}\xi^J$, $\Db \xi^{I}=\pb \xi^I+\Qb^{IJ} \xi^J$ with $\xi=\zeta$ or $\xi=\zetab$. (These covariant derivatives act only as normal derivatives when acting on $e^{\pm U/4}$.)

Let us now reason that, for this configuration, satisfying the projection
equations involving $P$, i.e.\ (\ref{delta_chi=0}), is necessary and
sufficient for a given amount of supersymmetry to be preserved on-shell. The integrability of the other supersymmetry equations
(\ref{delta_psi=0}) is assured if $[e^{U/4} D e^{-U/4}, e^{-U/4} \Db e^{U/4}]\zeta = 0$, which gives us the condition:
\be (\partial \overline{Q} - \overline{\partial}Q + Q\overline{Q} - \overline{Q}Q)\zeta = -\frac12 (\partial\overline{\partial} U) \zeta.\ee
On the right hand side, the expression is just $R_{z\overline{z}}/2$. From the
first Bianchi identity of (\ref{bianchi_id_form}), the left hand side is just
$-[P,\overline{P}]\zeta = -P^A\Gamma^{IJ}_{AB}\overline{P}^B
\zeta^J$. Using this and multiplying the projection equation
(\ref{delta_chi=0}) by $\Gamma_{\dot{A}B}^J$ to get the identity:\footnote{We wish to thank the anonymous referee for pointing out that this
equation can be used in an alternative proof of parts of Main Result \ref{thm:Pstructure}. Upon multiplication by $\overline{\zeta}^J$, one realizes that the equation essentially gives us the equation $[H,X] = 2X$ with $H\sim \zeta^I \overline{\zeta^J}$ and $X\sim P_z$, implying that $P_z$ is nilpotent and moreover giving an indication of which nilpotent orbits are allowed through an analysis of the stabilizer of $H$.}
:
\be 2 \zeta^I \Gamma^{IJ}_{A B} \overline{P}^A + \zeta^J \overline{P}^B =0,\ee
we can rewrite the integrability condition  as:
\be \frac12 \tr(P \overline{P})\zeta^I = \frac12 R_{z\overline{z}}
\zeta^I.\ee
The resulting equation is just the Einstein equation of motion
(\ref{eq:EOMtimelike}) and will always be satisfied on-shell, assuring
us that the other supersymmetry equations (\ref{delta_psi=0}) can be
integrated. The reverse reasoning can also be applied, i.e.\ if we have a solution $\zeta$ to the $Q$ equations, then using $[e^{U/4} D e^{-U/4}, e^{-U/4} \Db e^{U/4}]\zeta = 0$, the Bianchi identity, and the Einstein equation of motion, we get:
\be -P^A\Gamma^{IJ}_{AB}\overline{P}^B \zeta^J = \frac12 P^A \overline{P}^A \zeta^I,\ee
which can be multiplied by $\overline{\zeta}^I$ and rewritten as:
\be \| M^{\dot{A}} \|^2 = 0,\ee
where $\|\cdot\|$ is the complex vector norm and $M^{\dot{A}} = \Gamma^J_{\dot{A}B} \overline{P}^B \zeta^J$. It follows that $M^{\dot{A}} = 0$, which are exactly the $P$ projection equations. We can conclude that studying the $P$ projection equations (\ref{delta_chi=0}) is necessary and sufficient to guarantuee that a certain fraction of SUSY is preserved (on-shell): if we can pick a spinor at a point in spacetime, $\zeta(x_0)$, which satisfies the $P$ projection equations at that point, then we are guarantueed that the $Q$ equations can be integrated, i.e.\ we can extend $\zeta$ to a function over spacetime; but we are also assured that this resulting function $\zeta$ will satisfy the $P$ projection equations at every point in spacetime.\footnote{This is without taking into acount possible singular points where e.g.\ $P$ blows up (for instance, at points where singular brane sources sit), which may need to be excised from the spacetime.}

We have not discussed the precise boundary conditions for the fermions. The fermions transform both under the compact subgroup $K$ and as space-time spinors. If we are looking for solutions with monodromy $g_0\in G(\mathbb{Z})$, then $V=e^{\phi}$ will have the property that
\begin{equation}
V(e^{2\pi i} z) = g_0 V(z) K_0(z)
\end{equation}
and the natural boundary condition for the fermions is that they should
transform with $K_0(z)$ as
we go around the origin. Since
\begin{equation}
P(e^{2\pi i} z) = K_0(z)^{-1} P(z) K_0(z),
\end{equation}
this boundary condition is indeed compatible with the projection equations for unbroken supersymmetry generators (see (\ref{eq:Peqstatement}) and the discussion immediately after). We therefore apparently do not need to explicitly check the boundary conditions for the fermions and will not discuss this point explicitly in what follows.

\subsection{Lie Algebra Concepts}
\label{sec:lieconcepts}

We will be needing a few important concepts in the theory of Lie
algebras in our classification of supersymmetric solutions in 3D
(especially in the \emph{timelike} class). We introduce these concepts
here and illustrate them with simple $\la{sl}(n)$ examples.


\subsubsection{Root Decomposition of Lie Algebras}
\label{sec:lieconceptsroots} The root decomposition of Lie algebras is
well-known, but we review it here very quickly for reference as well as
discuss the root system of $\la{e}_8$.

In every Lie algebra $\la{g}$, we can select a number of commuting
(semi-simple) elements $H_i$. The maximum number of such elements that we can
select is called the \emph{rank} of $\la{g}$, and the collection
$\{H_i\}$ is called a Cartan subalgebra of $\la{g}$. Then, we can pick a
basis for the rest of $\la{g}$ that simultaneously diagonalizes all of
the generators $H_i$. This diagonalizing basis can be given by
$\{E_{\alpha}\}$, where $\alpha$ is a $\rank (\la{g})$-length vector
that denotes the eigenvalues of $E_{\alpha}$ under commutation with the
$H_i$'s, called a \emph{root} (vector). Roots have many properties,
e.g.\ if $\alpha$ is a root, then $-\alpha$ is also a root; the
collection of all vectors $\alpha$ that are roots is called a \emph{root
system}. In the end, all of the commutation relations of $\la{g}$ are
then given by:
\begin{align}
 [H_i, E_{\alpha}] &= \alpha_i E_{\alpha},\\
 [E_{\alpha}, E_{-\alpha}] &= \alpha\cdot H,\\
 [E_{\alpha}, E_{\beta}] &= N_{\alpha, \beta} E_{\alpha+\beta},
\end{align}
where $N_{\alpha,\beta}\neq 0$ only if $\alpha+\beta$ is a root. These
$N_{\alpha,\beta}$'s satisfy a number of consistency conditions (e.g.\
through the Jacobi identity), but there are still overall factors that
can be chosen arbitrary.

The Lie algebra $\la{e}_8$ has rank 8, so there are 8 generators in the
Cartan subalgebra; additionally there are 240 root generators
$E_{\alpha}$. The root system we will use consists of all 8-vectors with
two entries $\pm1$ and all other entries $0$ (112 such roots), and all
8-vectors with all entries $\pm \half$ with an even number of $+\half$ (128
such roots). Explicitly, we have all permutations of:
$(1,1,0,0,0,0,0,0)$, $(1,-1,0,0,0,0,0,0)$, $\half(1,1,1,1,1,1,1,1)$,
$\half(-1,-1,1,1,1,1,1,1)$, $\half(-1,-1,-1,-1,1,1,1,1)$, as well as their
negatives.

The Cartan involution $\theta$ has been mentioned above already. We repeat here that we always take it to act on the Lie algebra as:
\be \theta(E_{\alpha}) = -E_{-\alpha}, \qquad \theta(H_i) = -H_i,\ee
so that the compact generators (those with eigenvalue $+1$ under $\theta$) are given by $E_{\alpha}-E_{-\alpha}$, and the non-compact generators (those with eigenvalue $-1$ under $\theta$) are given by $H_i, E_{\alpha}+E_{-\alpha}$.

\subsubsection{The Adjoint Representation; Nilpotent and Semi-simple Elements}
The well-known \emph{adjoint representation} of a Lie algebra takes an element to its action via the commutator:
\be \adj: \quad\la{g} \mapsto \End(\la{g}): \, X \rightarrow \ad_X;\qquad \ad_X(Y) = [X,Y].\ee
Using the adjoint representation, we can divide elements of the Lie
algebra into three distinct classes (the only overlap between the three
classes is the element $0$):
\begin{itemize}
 \item \emph{Nilpotent} elements: elements $X$ for which $(\ad_X)^n
       \equiv 0$ for some finite $n$.
       \\
       The \emph{Jacobson-Morozov}
       theorem tells us that for any nilpotent element $X$ in an algebra
       $\la{g}$, we can always find elements $H,Y\in \la{g}$
       which satisfy
       \begin{align}
        [H,X]=2X,\qquad [H,Y]=-2Y,\qquad [X,Y]=H.
       \end{align}       
       Such a triple
       $\{H,X,Y\}$ is called a \emph{standard triple} and it generates
       an $\la{sl}(2)$ subalgebra of $\la{g}$.\\ In
       $\la{sl}(n)$, typical examples of nilpotent elements are the upper-triangular
       matrices (with zero on the diagonal).
\item \emph{Semi-simple} elements: elements $X$ for which $\ad_X$ is
      diagonalizable (in the complexified version of the Lie algebra).\\ In $\la{sl}(n)$, obvious examples are any
      traceless diagonal matrix.
\item \emph{Other} elements: elements that are neither nilpotent nor
      semi-simple. For all elements $X$, the unique Jordan decomposition
      is given by: \be X = X_S + X_N,\ee where $X_S$ is semi-simple,
      $X_N$ is nilpotent, and $[X_S,X_N]=0$. Lie algebra theory
      guarantees that such a unique splitting always exists and $X_S, X_N$ lie
      in the same algebra as $X$.\\ An interesting fact is that all
      elements of $\la{sl}(2)$ are either nilpotent or semi-simple
      \cite{collingwood}. So the ``smallest'' algebra where we can find
      such a non-semi-simple and non-nilpotent element is
      $\la{sl}(3)$; an example is $X=\left(\begin{array}{ccc} 1 & 1
      & 0\\ 0 & 1 & 0\\ 0&0&-2\end{array}\right)$, where $X_S$ is the
      diagonal part and $X_N$ the off-diagonal part.
\end{itemize}
For semisimple Lie algebras, $X$ is nilpotent (resp. semisimple) if and only if $X$ is nilpotent (resp. semisimple) in every finite dimensional representation of the Lie algebra. It follows that an embedding of an algebra into another algebra preserves nilpotency or semi-simplicity of each element \cite{collingwood}. One corrollary that this implies is that if an element $X$ is a part of an $\la{sl}(2)$ subalgebra of a given Lie algebra $\la{g}$, then $X$ must be either semi-simple or nilpotent.

\subsubsection{Orbits; Topology of Orbits}
\label{sec:lieconceptsorbits}
Another very important concept in Lie algebras is that of an \emph{orbit} $\mathcal{O}_X$ of an element $X$, also known as its \emph{conjugacy class}. We refer to the excellent book \cite{collingwood} for more details on orbits (especially nilpotent orbits) in a Lie algebra; we will only give an intuitive sketch of the subject here. Unless otherwise specified, the results mentioned are valid both for orbits in real and complex Lie algebras.

For an element $X$ in a matrix Lie algebra $\la{g}$ with associated connected matrix Lie group $G$, the conjugacy class or orbit is defined naturally as:
\be \mathcal{O}_X = G \cdot X = \{ M\cdot X \cdot M^{-1}, \; M\in G \}.\ee
There is a natural extension of this definition to define orbits in non-matrix Lie algebras as well using the natural action of the Lie group $G$ on the Lie algebra $\la{g}$ \cite{collingwood}.

It is very important to realize that orbits always only contain one type of elements (nilpotent, semi-simple, or other); thus, the study of e.g.\ classifying all nilpotent elements in a Lie algebra can be reduced to the somewhat easier problem of classifying all nilpotent \emph{orbits} in a Lie algebra.

Let us collect a few important and interesting facts on the orbits of the semi-simple and nilpotent types (not much can be said about the other type):
\begin{itemize}
 \item \emph{Semi-simple orbits.} There are infinitely many of these orbits. For complex Lie algebras, every semi-simple orbit contains exactly one element of a given Cartan subalgebra (up to Weyl reflections). In other words, every semi-simple element in a complex Lie algebra is conjugate to an element in a given Cartan subalgebra. In particular, the well-known result that all Cartan subalgebras are conjugate follows from this.\\
For real Lie algebras, the same is not true. As an example, consider the Cartan subalgebras generated by elements of the form $\left(\begin{array}{cc} \lambda&0\\0&-\lambda \end{array}\right)$ in $\la{sl}(2)$. This Cartan subalgebra is complex conjugate to the Cartan subalgebra generated by $\left(\begin{array}{cc} 0&\lambda\\-\lambda&0 \end{array}\right)$ with $\lambda\in\mathbb{C}$, but these are not two conjugate subalgebras in $\la{sl}(2)_{\mathbb{R}}$ where we take $\lambda\in\mathbb{R}$ in both subalgebras.\\
For the real algebra $\la{e}_{8(8)}$ (the case we are ultimately interested in), there are 10 distinct conjugacy classes of Cartan subalgebras \cite{sugiura59}.
\item \emph{Nilpotent orbits.} One surprising fact is that there are
       only finitely many nilpotent orbits in any (real or complex) Lie
       algebra. They can and have been enumerated for all semi-simple
       complex and real Lie algebras \cite{collingwood}. In general, the
       intersection of a complex nilpotent orbit with the real algebra
       is a union of multiple real nilpotent orbits.\\
For example, in $\la{sl}(2)_\bbC$, there is only one nilpotent orbit, generated by the element $\left(\begin{array}{cc} 0&1\\0&0\end{array}\right)$. However, in $\la{sl}(2)_\bbR$, there is also a second nilpotent orbit generated by $\left(\begin{array}{cc} 0&-1\\0&0\end{array}\right)$.
\end{itemize}

There is a natural topology which one can impose on a Lie algebra called the Zariski topology. The closed sets in the Zariski topology are the sets of zeros of polynomials in the elements. We will not spend too much details on the specifics here (although we will need this topology explicitly in e.g.\ the proof of Result \ref{prop:Pclosure}); however, a number of results concerning the topology of orbits will be important in the following.\\
\emph{Note: The closure $\overline{\mathcal{O}}$ of an orbit $\mathcal{O}$ should not be confused with the complex conjugate $\overline{P}$ of an element $P$!}

First of all, let us stress that the closure of an orbit is always a union of orbits (because the closure of an orbit is $G$-invariant). Further, we can collect some important facts about the three types of orbits:
\begin{itemize}
 \item \emph{Semi-simple orbits.} These are the only orbits that are closed sets by themselves. This means that for a semi-simple $X$, $\overline{\mathcal{O}}_X = \mathcal{O}_X$. This also implies that every closure of an orbit contains a semi-simple orbit. The zero orbit $\mathcal{O}_0=\{0\}$ is by convention the only orbit which is both semi-simple and nilpotent.
\item \emph{Other orbits.} An element $X$ with unique Jordan decomposition $X=X_S+X_N$ has the important property that $\mathcal{O}_{X_S}\subset\overline{\mathcal{O}}_X$.
\item \emph{Nilpotent orbits.} These orbits are the only orbits that
contain the zero orbit, $\mathcal{O}_0=\{0\}$, in their
closure. Further, it is very non-trivial that there exists a
well-defined partial ordering structure on nilpotent orbits as follows:
$\mathcal{O}_i\leq \mathcal{O}_j$ if $\overline{\mathcal{O}}_i
\subseteq\overline{\mathcal{O}}_j$. This partial ordering allows us to
draw a so-called Hasse diagram, where orbits are ordered from left to
right by dimension of the orbit and a line is drawn between two orbits
if one is contained in the closure of the other. The (partial) Hasse
diagram for nilpotent orbits in $\la{e}_{8(8)}$ is given in Fig.\
\ref{fig:hassediagram} in appendix \ref{sec:theorems}\@. Note how e.g.\
$\mathcal{O}_1\leq \mathcal{O}_3$ as we can trace a line backwards from
$\mathcal{O}_3$ to $\mathcal{O}_1$.
\end{itemize}

Finally, there is one more type of orbit that we will be interested
in. Take a Cartan
decomposition $\la{g}=\la{k}\oplus\la{p}$ of $\la{g}$ as in (\ref{eq:cartandecomp}), i.e.\ when
$\la{g}$ is viewed as a real Lie algebra, $\la{k}$ contains
compact generators and $\la{p}$ non-compact generators; and call
$G,K$ the Lie groups that are associated to, respectively,
$\la{g},\la{k}$. Then we will also be interested in
complex $K(\bbC)$-orbits in $\la{p}_\bbC$. Note that
these objects are well-defined, as the action of $K$ on $\la{p}$ is
internal. We will see later on that these orbits are crucial in the
characterization of the \emph{timelike} class of supersymmetric
solutions.

Up until now, we have only discussed complex and real $G$-orbits in
$\la{g}$; one might \emph{a priori} think that these $K$-orbits in
$\la{p}$ are completely different animals. Thankfully, there is the
\emph{Kostant-Sekiguchi bijection} which is a natural one-to-one
correspondence between $G(\bbR)$-orbits in $\la{g}_\bbR$ and
$K(\bbC)$-orbits in $\la{p}_\bbC$ \cite{collingwood}.\footnote{Note that this correspondence is between orbits and not between elements in the orbits.} This natural
bijection also carries over the topology structure of the
orbits.\footnote{At least for classical Lie algebras \cite{collingwood}
and for $E_{8(8)}$ \cite{Djokovic:hasse03}; presumably this is true for
other exceptional Lie algebras as well.} Thus, when we talk about
$K(\bbC)$-orbits in $\la{p}_\bbC$, we will be able to use many
results concerning $G(\bbR)$-orbits in $\la{g}_\bbR$. A few other
results concerning $K(\bbC)$-orbits in $\la{p}_\bbC$ that are
needed are collected in appendix \ref{sec:theorems}\@. Sometimes, we will want to relate an element of a complex $K$-orbit in $\la{p}$ to an element in the corresponding (through the Kostant-Sekiguchi bijection) real $G$-orbit in $\la{g}$ (or vice versa). We will denote such a relationship as:
\be \label{eq:defKSbij} P_z \bowtie X,\ee
which should be read as ``the complex $K$-orbit of $P_z$ corresponds to the real $G$-orbit of $X$ through the Kostant-Sekiguchi bijection''. Intuitively, one can think of this as being a kind of generalization of the concept of ``conjugate to'', since the elements in orbits connected by this bijection share many important properties. 

For $\la{g}_\bbR=\la{e}_{8(8)}$, there are 116 nilpotent
$G(\bbR)$-orbits in $\la{g}_\bbR$ (and infinitely many
semi-simple and other orbits, as explained above); these are typically
labelled $\mathbf{0}$--$\mathbf{115}$, where $\mathbf{0}=\{0\}$ is the
trivial orbit. Because of the Kostant-Sekiguchi bijection, this same
numbering applies to the nilpotent $K(\bbC)$-orbits in
$\la{p}_\bbC$.

\subsubsection{Standard Triples; Cayley Triples; Cayley Transforms}
\label{sec:lieconceptstriples}
A few times in this paper, we will be referring to a ``standard triple''. A \emph{standard triple} is a triple $\{H,X,Y\}$ in Lie algebra $\la{g}$ which satisfies:\footnote{The conventions of the standard triple sometimes differ in the sign of the third commutator. We follow the convention of e.g.\ \cite{collingwood}, while e.g.\ \cite{Djokovic:reps} has an extra minus sign in this commutator.}
\be [H,X]= 2X; \qquad [H,Y]=-2Y; \qquad [X,Y] = H.\ee
Clearly, $\mathrm{span}\{H,X,Y\}$ is an $\la{sl}(2)$-subalgebra of $\la{g}$.

Given a Cartan involution $\theta$, a Cayley triple is defined as a standard triple which satisfies:
\be \theta(X)=-Y; \qquad \theta(Y)=-X; \qquad \theta(H)=-H.\ee
For example, for any root $\alpha$ and the Cartan involution as specified at the end of section \ref{sec:lieconceptsroots}, $\{\alpha\cdot H, E_{\alpha}, E_{-\alpha}\}$ is such a Cayley triple.

Given a Cayley triple, the Cayley transform of the Cayley triple is given by the standard triple $\{H',X',Y'\}$, with:
\be X' = \frac12(X+Y + iH); \qquad Y' = \frac12(X+Y-iH); \qquad H' =i(X-Y).\ee
It is easy to see that this new triple consists of eigenvectors of $\theta$, i.e.\ $\theta(X') = -X'$, $\theta(Y')=-Y'$; $\theta(H')=-H'$. (A triple with this property is also called a \emph{normal} triple.)

As mentioned above, the Jacobson-Morozov theorem guarantuees us that
every nilpotent element can be embedded in a standard triple. This is
crucial in the study and classification of complex and real nilpotent
orbits, which relies heavily on this standard triple and in particular
the neutral element $H$ that accomponies the nilpotent $X$. Any
nilpotent $G(\bbR)$-orbit in $\la{g}_{\bbR}$ contains a representative
$X$ which is the nilpositive element of a real Cayley triple. The
Kostant-Sekiguchi bijection between $G(\bbR)$-orbits in $\la{g}_\bbR$
and $K(\bbC)$-orbits in $\la{p}_\bbC$ sends the orbit through this
nilpositive element $X$ (in a real orbit) of a real Cayley triple into
the complex orbit of its Cayley transform $X'$, i.e., $X\bowtie X'$ using
the notation defined above.


\section{Structure of SUSY solutions}\label{sec:allresults}
In this section, we will discuss our results on characterizing
supersymmetric solutions in 3D maximal supergravity as described
above. This section will always use the notation
$G(\bbR)\equiv E_{8(8)};\ G(\bbC)\equiv E_8(\bbC);\ 
K(\bbC/\bbR)\equiv SO(16,\bbC/\bbR);\ 
\la{k}_{\bbC/\bbR} \equiv \la{so}(16)_{\bbC/\bbR};\
\la{p}_\bbR \equiv \la{e}_{8,8}\ominus\la{so}(16)_\bbR;\ 
\la{p}_\bbC\equiv (\la{e}_{8})_\bbC\ominus\la{so}(16)_\bbC$.

\subsection{Statement of Results}\label{sec:statementresults}
Our first result is about the spacetime metric of a SUSY solution:
\begin{theorem}\label{thm:spacetime}
A spacetime that preserves some supersymmetry is of one of two classes:
\begin{itemize}
\item The \emph{timelike} class, for which the metric is static and can
      always be brought to the form:
\be \label{eq:timelikemetric} ds^2 = -dt^2 + e^{U(z,\bar{z})}( dz d\bar{z}).\ee
In these spacetimes, $P_t=0$; the scalars do not depend on $t$.
 \item The \emph{null} class of \emph{pp-waves}, for which the metric can always be brought to the form:
\be \label{eq:nullmetric} ds^2 = -2du dv - 2\omega(v,x) dv dx + h(v,x)dx^2.\ee
The scalars only depend on $v$; the only non-zero component of $P$ is $P_v$.
\end{itemize}
\end{theorem}

Our second main result is to classify the \emph{timelike} class of solutions and will feature the quantity $P_z$ defined above in section \ref{sec:maxsugra} prominently. As we saw above at the end of section \ref{sec:maxsugra}, satisfying the projection equations  (\ref{delta_chi=0}) involving $P_z$ is a necessary and sufficient condition for preserving a given amount of supersymmetry, as the other supersymmetry equations (\ref{delta_psi=0}) are guaranteed to be integrable on-shell. This $P_z$ is an element of the complexified coset algebra, $\la{p}_\bbC$. The equations we are studying are the projection equations:
\be \label{eq:Peqstatement} P_z^A \Gamma^I_{A\dot{A}} \overline{\zeta}^I = 0. \ee
We want to investigate for which $P^A$ there can be non-trivial complex spinor solutions $\zeta^I$ to these equations.

It is immediately clear that only the conjugacy class of $P_z$, called its \emph{orbit}, is important: the existence of non-trivial solutions $\zeta^I$ is related to the rank of the matrix $M^{\dot{A}}_{\ I}=P^A \Gamma^I_{A\dot{A}}$ (see especially the proof of result \ref{prop:Pclosure} below), and this rank is unaffected by conjugating $P^A$ with any group element of $K(\bbC)$.

So we are interested in $K(\bbC)$-conjugacy classes, i.e.\ $K(\bbC)$-orbits in $\la{p}_\bbC$. As explained above at the end of section \ref{sec:lieconceptsorbits}, there is a natural bijection between these orbits and $G(\bbR)$-orbits in $\la{g}_\bbR$ \cite{collingwood,Djokovic:hasse03} which provides many results on the structure of these orbits. There are also some additional results for these orbits \cite{KostantRallis71} that we will need; these are collected in appendix \ref{sec:theorems} with the appropriate references there; in the proofs of our results below we will reference these theorems in the appendix as needed.

The main fact needed to understand our result below (also discussed at the end of section \ref{sec:lieconceptsorbits}) is that there are 116 (i.e.\ a finite number of) non-trivial nilpotent $K(\bbC)$-orbits in $\la{p}_\bbC$, which are all numbered as \textbf{0} through \textbf{116} \cite{collingwood, Djokovic:hasse03}.

Our classification for the \emph{timelike} class of supersymmetric solutions can then be summarized as follows:
\begin{theorem}\label{thm:Pstructure}
For a supersymmetric solution of the \emph{timelike} class, the quantity $P_z$, at every point in spacetime, is a (nilpotent) element of one of the ten nilpotent orbits in $\la{p}_\bbC$ labelled \textbf{0}, \textbf{1}, \textbf{2}, \textbf{3},\textbf{ 4}, \textbf{6}, \textbf{7}, \textbf{9}, \textbf{12}, or \textbf{14}. Which of these orbits that $P_z$ sits in determines how much supersymmetry is preserved (see Table \ref{tab:orbits}).
\end{theorem} 
The concepts of nilpotency and orbits are introduced in section \ref{sec:lieconcepts}.

Main Result \ref{thm:Pstructure} can be obtained by splitting it into three smaller results:
\begin{enumerate}[(a)]
 \item \label{prop:Pclosure} \emph{Say a $P_z$ in orbit $\mathcal{O}_1$ preserves $x$ supercharges, and a $P_z$ in orbit $\mathcal{O}_2$ preserves $y$ supercharges. Then $\overline{\mathcal{O}_1} \subseteq \overline{\mathcal{O}_2}$ implies $x\geq y$.}
\item \label{prop:Pnilpotent}
\emph{For all spacetimes preserving supersymmetry, $P_z$ must be a nilpotent element in $\la{p}_\bbC$ (at every point in spacetime).}
\item \label{prop:Porbits}
\emph{If $P_z$ is in nilpotent orbits \textbf{0}, \textbf{1}, \textbf{2}, \textbf{3},\textbf{ 4}, \textbf{6}, \textbf{7}, \textbf{9}, \textbf{12}, or \textbf{14}, then the spacetime can preserve supersymmetry. If $P_z$ is in any other nilpotent orbit, the spacetime always breaks all supersymmetry. The amount of supersymmetry preserved by $P_z$ depends on which orbit $P_z$ sits in (see Table \ref{tab:orbits}).}
\end{enumerate}

Given these three results, it is easy to see how the Main Result \ref{thm:Pstructure} is obtained: for spacetimes in the \emph{timelike} class that preserve supersymmetry, result \ref{prop:Pnilpotent} tells us that $P_z$ must be nilpotent, while result \ref{prop:Porbits} tells us which nilpotent orbits $P_z$ can be in.

A few remarks are in order about this timelike class:
\begin{itemize}
 \item All of the orbits that preserve some supersymmetry are not only nilpotent but are of the type $n\times \la{sl}(2)$ as is apparent from the `Structure' column in Table \ref{tab:orbits}. The `Structure' column lists the type of minimal regular (i.e.\ normalized by a Cartan subalgebra of $\la{g}$) subalgebra $\la{s}\subseteq\la{e}_8$ that meets the nilpotent orbit $\mathcal{O}$. (See \cite{Djokovic:reps} for the explicit precise definition.) Intuitively (and loosely speaking), nilpotent elements in such orbits with structure $n\times A_1$ can be seen as constructed as sums of nilpotent elements from commuting $\la{sl}(2)$ subalgebras.
 \item We note that a classification along the lines of table VIII in \cite{Bossard:2009at} should be possible, where the nilpotent orbit of $P_z$ is completely determined by its properties in various representations. Since we do not need this kind of classification here, we will just note that in the adjoint representation (where only the trivial orbit $\mathbf{0}$ satisfies $\left(\textrm{ad}_{P_z}\right)^2 = 0$), we have $\left(\textrm{ad}_{P_z}\right)^3 = 0$ for orbits \textbf{1} and \textbf{2}; $\left(\textrm{ad}_{P_z}\right)^4 = 0$ for orbits \textbf{3} and \textbf{6}; and $\left(\textrm{ad}_{P_z}\right)^5 = 0$ for orbits \textbf{4}, \textbf{7}, \textbf{9}, \textbf{12}, and \textbf{14}. However, also some non-supersymmetric nilpotent orbits have $\left(\textrm{ad}_{P_z}\right)^5 = 0$ - for example, orbit \textbf{5}.
 \item Although our results completely classify supersymmetric solutions, this classification is given in terms of $P_z$. It is unclear what the physical meaning of $P_z$ is, and what $P_z$ being in a particular orbit means for restricting possible solutions. Investigating this is the subject of sections \ref{sec:simplebranes} and \ref{sec:morebranes}, where we find a plethora of explicit single-center brane solutions.
\end{itemize}

Finally, we can completely specify the \emph{null} class of solutions in a relatively straightforward way:
\begin{theorem}\label{thm:nullclass}
A supersymmetric solution of the \emph{null} class, with a metric as given by (\ref{eq:nullmetric}), preserves 1/2 of the supersymmetries. The scalars (which are only functions of $v$ as stated above) are related to the metric functions as:
\be \label{eq:nullclassEOM} \frac{(\partial_v h)^2 +2\partial_x h\partial_v\omega-2h(2\partial_x\partial_v\omega+\partial_v^2 h)}{4h^2} =  P_v^A(v) P_v^A(v).\ee
Further, if the equation:
\be \label{eq:nullfcondition} -\omega +\partial_x f + c\partial_v\omega - h\frac{\partial^2_v f}{\partial_v\omega}+h\frac{\partial_v f\partial^2_v\omega}{(\partial_v\omega)^2}=0,\ee
has a non-trivial solution for the function $f(x,v)$ and constant $c$, then the metric can always be put by a coordinate transformation in the simple form:
\be \label{eq:nullmetricnew} ds^2 = -2du dv + h_0(v)^2 dx^2.\ee
In this case, equation (\ref{eq:nullclassEOM}) becomes:
\be \label{eq:nullclassEOMnew} h_0'' +  h_0 P_v^A(v) P_v^A(v)=0.\ee
\end{theorem}


\subsection{Derivation of Results}\label{sec:derivationresults}

\subsubsection{Two Classes of Spacetimes}
\begin{proof}[Proof of Main Result \ref{thm:spacetime}]
Consider the covariant vector
\begin{align}
 V_\mu\equiv (\epsilon^I)^T\gamma_{\hat{0}}\gamma_\mu \epsilon^I.
\end{align}
By requiring that the supersymmetry variation for gravitino
\eqref{susy_var_psi} vanish, we can easily derive
\begin{align}
 \nabla_\nu V_\mu=0\label{nabla_V=0}
\end{align}
with $\nabla_\nu$ the standard covariant derivative in three dimensions.
Note that in the expression for $V_\mu$ the spinors are ordinary numbers
and therefore commute with each other.  If we choose coordinates at a
point so that the metric at that point is Minkowski, the norm of $V_\mu$
becomes
\begin{align}
 |V|^2=-4(\epsilon_1\cdot \epsilon_1)(\epsilon_2\cdot \epsilon_2)
 +4(\epsilon_1 \cdot\epsilon_2)^2
\end{align}
where $\epsilon_1\cdot \epsilon_2\equiv \epsilon_1^I\epsilon_2^I$ etc.
Therefore, there are two possibilities, either $\epsilon_1$ is parallel to $\epsilon_2$ and $V$ is null, or $\epsilon_1$
and $\epsilon_2$ are not parallel and $V$ is timelike.

\begin{itemize}
 \item
      \emph{$V$ timelike:}

If $V$ is timelike, then the condition \eqref{nabla_V=0} implies that $V$ is Killing, but is
actually stronger than that. If we locally choose coordinates $(t,x^1,x^2)$
such that $V=\partial_t$, the metric will not
depend on $t$ as a consequence of the fact that $V$ is Killing:
\be ds^2 = -f(x)(dt + A_i(x)dx^i)^2 + h_{ij}(x)dx^i dx^j.\ee
However, the condition \eqref{nabla_V=0} is much stronger actually implies that $f$ is a constant and $A_i(x)=\partial_i A(x)$, which means that by a redefinition of the coordinate $t\rightarrow |f|^{1/2}(t-A(x))$, we can bring the metric into the form:
\begin{align}
 ds^2 = - dt^2 + h_{ij}(x)dx^i dx^j, \qquad i=1,2,
\end{align}
in other words it must at least locally be a direct product. By choosing conformal gauge on the 2d surface, the metric becomes
\begin{align}
 ds^2 = -dt^2 + e^{U(x)} [(dx^1 )^2+ (dx^2 ) ^2 ] = -dt^2 + e^{U(z,\zb)} dzd\zb,
\end{align}
where in the last equality we introduced complex coordinates $z=x^1+i x^2$.

From the explicit form of $V_\mu$ , we can then also derive that
\begin{align}
 \epsilon_1\cdot \epsilon_1=\epsilon_2\cdot \epsilon_2={\rm const.},
 \qquad
 \epsilon_1\cdot \epsilon_2=0.
\end{align}

It is also easy to see that the scalar fields have to be time
independent. Indeed, the $(t,t)$ component of the Einstein equation
\eqref{einstein_eq} reads
\begin{align}
 P^A_t P^A_t=R_{tt}=0
\end{align}
which shows that $P^A_t=0$, and this in turn can only happen when $\partial_t \phi=0$.

 \item \emph{$V$ null:}

If $V$ is null, we can locally choose coordinates $(u,v,x)$ such that $V=\partial_u$. Using a similar reasoning as in \cite{Gutowski:2003rg}, the general form of such a metric is given by:
\be ds^2=-2 (H dv+\gamma dx)\left[ du +\omega dx + G(Hdv+\gamma dx)\right]+hdx^2,\ee
where $H,\gamma,\omega,G, h$ depend on $x,v$. Now, (\ref{nabla_V=0}) implies that $H=\partial_v K, \gamma=\partial_x K$. Using $K(x,v)$ as a coordinate and shifting $u$ by a function of $(x,v)$, we can then bring the metric to the form (\ref{eq:nullmetric}).

Two of the Einstein equations are now $R_{xx}=R_{uu}=0$ so that, using the same reasoning as above in the \emph{timelike} case, we have $P_u=P_x=0$ and $\partial_x\phi=\partial_u\phi=0$. The only non-zero component of $P$ is then $P_v(v)$.
\end{itemize}
\end{proof}

\subsubsection{Closure relation}
\begin{proof}[Proof of result \ref{prop:Pclosure}]\footnote{We thank M.~Baggio for discussions on this proof.}
We want to prove the statement that an orbit $\mathcal{O}'$ contained in the closure $\overline{\mathcal{O}}$ of $\mathcal{O}$ must preserve at least the amount of SUSY that $\mathcal{O}$ does.\\
The condition to preserve SUSY is captured in the (complexified) equations:
\be P^A \Gamma^I_{A\dot{A}} \overline{\zeta}^I = 0.\ee
For a given $P^A$, these are 128 linear homogeneous equations which the 16 components of $\zeta^I$ must satisfy. The rank of the matrix $M^{\dot{A}}_{\ I}=P^A \Gamma^I_{A\dot{A}}$ is crucial for the amount of SUSY preserved.\\
The rank-nullity theorem tells us that the matrix rank of $M$ plus its nullity space (i.e.\ the dimension of the space of vectors $\zeta^I$ which it annihilates) must add up to 16. Demanding that we preserve at least a fraction $f$ of SUSY is equivalent with demanding at least $16f$ linear independent Killing spinors and thus equivalent with demanding that the rank of $M$ is at most $16(1-f)$. If the rank of $M$ is at most $16(1-f)$, then we want all $(16(1-f)+1)\times(16(1-f)+1)$ submatrices of $M$ to have vanishing determinant. These determinant equations are of the form $D_i=0$, where $D_i$ is a homogeneous polynomial in the components of $P_z$. The set of zeros $Z(D_i)$ of the polynomial $D_i$ is a closed set in the Zariski topology. Then $C_f\equiv \bigcap_i Z(D_i)$ (where $i$ runs over all determinants of $(16(1-f)+1)\times(16(1-f)+1)$ submatrices of $M$) is an intersection of closed sets, so is closed. Clearly, $C_f$ is set of $P$'s preserving at least a fraction $f$ of SUSY\@. Elementary determinant theory implies that $C_f\subset C_{f'}$ if and only if $f\geq f'$. This implies the wanted relationship, because if $\mathcal{O}\subset C_f$, then also $\overline{\mathcal{O}}\subset C_f$ and then for every $\mathcal{O}'\subset \overline{\mathcal{O}}$ we have $\mathcal{O}'\subset C_f$, so $\mathcal{O}'$ preserves at least $f$ SUSY\@.
\end{proof}

\subsubsection{$P_z$ nilpotent}\label{sec:proofPnilpotent}
\begin{proof}[Proof of result \ref{prop:Pnilpotent}]
We will prove that $P$ must be nilpotent at every (space-time) point $x$ if the spacetime is to preserve any SUSY\@.

Assume that $P_z$, at some point $x$, is not nilpotent, but still preserves some SUSY\@. Then, we know that (all elements in) the orbit $\mathcal{O}_P$ (by this we mean the $K(\bbC)$ orbit of $P_z$ in $\la{p}_\bbC$) preserve SUSY, as well as all elements in its closure $\overline{\mathcal{O}_P}$ (from result \ref{prop:Pclosure}). Now, we rely on the following statement (which we will prove in a moment):\\
\emph{There exists a non-zero element $c_i H_i$ (i.e.\ an element in the CSA which we have constructed above, i.e.\ with all $H_i\in \la{p}_\bbC$) for which $c_i H_i \in \overline{\mathcal{O}_P}$.}\\
However, it is easy to prove, by a simple brute-force equation solving in Mathematica (see appendix \ref{sec:appconstruction}, especially \ref{sec:appproofPnilpotent}, for more details), that the only element $c_i H_i$ which preserves SUSY (i.e.\ for which there is a non-trivial Killing spinor satisfying the SUSY equation for $P_z$) is the one where $c_i=0$ for all $i$. The lack of non-zero $c_iH_i$ preserving SUSY means $P_z$ must also break SUSY completely. Thus, we are finished: the only elements $P_z$ which can possibly preserve SUSY are nilpotent!

We still need to prove the italic statement above. First of all, if $P_z$ is not nilpotent, then there exists a unique Jordan decomposition of $P_z=P_S + P_N$ with $P_S$ semi-simple and non-zero, $P_N$ nilpotent, and $[P_S,P_N]=0$. Since $P_z\in \la{p}_\bbC$, we have $P_S, P_N\in \la{p}_\bbC$ (lemma \ref{lemma:KRxinp}). Now, we have that $P_S\in \overline{\mathcal{O}_P}$ (lemma \ref{lemma:KRxsinclosure}) so from this also follows trivially that for the orbit  $\mathcal{O}_{P_S}$ of $P_S$, we have $\mathcal{O}_{P_S} \subset \overline{\mathcal{O}_P}$ (because $\overline{\mathcal{O}_P}$ is $K$-stable, i.e.\ it must consist of a union of orbits). But now we have that every non-zero semi-simple orbit contains a non-zero element of the CSA spanned by $\{H_i\}$ (lemma \ref{lemma:KRxsconjugatecartan}), so we conclude that there exists a non-zero element in the CSA, $c_i H_i$ (with complex $c_i$'s), for which $c_i H_i \in \overline{\mathcal{O}_P}$.
\end{proof}

\subsubsection{Table of SUSY orbits}\label{sec:proofPorbits}
\begin{proof}[Proof of result \ref{prop:Porbits}]
We can verify by direct calculation (see appendix \ref{sec:appconstruction}, especially \ref{sec:appproofPorbits}, for more details) that a $P_z$ in nilpotent orbit \textbf{0}, \textbf{1}, \textbf{2}, \textbf{3},\textbf{ 4}, \textbf{6}, \textbf{7}, \textbf{9}, \textbf{12}, and \textbf{14} preserves the amount of SUSY as given in Table \ref{tab:orbits}. In the same way, we can verify that $P_z$ in the nilpotent orbit \textbf{5} breaks all SUSY\@.

Lemma \ref{lemma:E8hasse} tells us that all nilpotent orbits, with the exception of the ten SUSY-preserving orbits \textbf{0}, \textbf{1}, \textbf{2}, \textbf{3},\textbf{ 4}, \textbf{6}, \textbf{7}, \textbf{9}, \textbf{12}, and \textbf{14}, contain orbit $\textbf{5}$ in their closure. Using Result \ref{prop:Pclosure}, it then immediately follows that all orbits except the nine mentioned above must break all SUSY\@.
\end{proof}

\begin{landscape}
\begin{table}
\begin{center}
 \begin{tabular}{|l||l||r|r||l||l|}
 \hline
 $\mathbb{R}$ label & Structure & $\dim$ & $\inv$ & Representative & SUSY\\
 \hline\hline
 \textbf{0}  & $0$ & 0 & 120 & (vacuum) & 1\\
 \hline\hline
 \textbf{1}  & $A_1$ & 58 & 64 & $D7(1234567)$ & 1/2\\
 \hline\hline
 \textbf{2} & $2A_1$ & 92 & 44 & $D3(123)+D7(1234567)$ & 1/4\\
 \hline\hline
 \textbf{3}   & $3A_1$ & 112 & 40 & $P(1)+D3(123)+D7(1234567)$ & 1/8\\
 \hline\hline
 \textbf{4}  & $(4A_1)''$ & 114 & 70 & $D5(34567)+D5(12567)+D5(12347)+D1(7)$ & 1/8\\
 \hline
 5   & $(4A_1)'', A_2$ & 114 & 64 & $D1(1)+D7(1234567)$ & 0\\
 \hline\hline
 \textbf{6}  & $(4A_1)'$ & 128 & 32 & $D3(125)+D3(234)+D3(136)+D1(7)$ & 1/16\\
 \hline\hline
 \textbf{7}  & $5A_1$ & 136 & 38 & (orbit 4) $+P(7)$ & 1/16\\
 \hline
 8  & $5A_1,A_2+A_1$ & 136 & 32 & (orbit 5) $+P(7)$ & 0\\
 \hline\hline
 \textbf{9}  & $6A_1$ & 146 & 38 & (orbit 7) $+KK(34567;2)$ & 1/16\\
 \hline
 10   & $6A_1,A_2+2A_1$ & 146 & 26 & (orbit 8) + $KK(34567;2)$ & 0\\
 \hline\hline
 11  & $A_3$ & 148 & 32 & $D1(7)+D5(12347)+KK(12345;6)$ & 0\\
 \hline\hline
 \textbf{12}   & $7A_1$ & 154 & 50 & (orbit 9) $+KK(12567;3)$ & 1/16\\
 \hline
 13  & $7A_1,A_2+3A_1$ & 154 & 26 & (orbit 10) + $KK(12567;4)$ & 0\\
 \hline\hline
 \textbf{14}  & $8A_1$ & 156 & 92 & (orbit 12) $+KK(12347;6)$ & 1/16\\
 \hline
 15 & $8A_1,A_2+4A_1$ & 156 & 50 & (orbit 13) + $KK(12347;5)$ & 0\\
 \hline
 16  & $8A_1,A_2+4A_1,2A_2$ & 156 & 44 & $D1(7)+D5(12345)+D5(34567)+KK(14567;2)$ & 0\\ 
 \hline\hline
 17  & $2A_2+A_1$ & 162 & 24 & (orbit 16) $+D5(12367)$ & 0\\
 \hline
 \end{tabular}
 \caption{\label{tab:orbits}Tabulation of all (nilpotent) orbits preserving some supersymmetry as well as a few that do not. See section \ref{sec:proofPorbits} about the SUSY-preservation of the orbits. See section \ref{sec:branereps} on how the orbit representative was found. See the end of appendix \ref{sec:appbranerep} on the characterization of the nilpotent orbit by its dimension $\dim$ and its orbit invariant $\inv$. Note that the orbits grouped together (for example, orbits \textbf{4} and \textbf{5}) are those that, as $G(\mathbb{R})$-orbits in $\mathbf{g}_{\mathbb{R}}$, are in the same $G(\mathbb{C})$-orbit in $\mathbf{g}_{\mathbb{C}}$; we do not need this property of the orbits in this paper.}
\end{center}
\end{table}
\end{landscape}

\subsubsection{Specifications of Null Class}
\begin{proof}[Proof of Main Result \ref{thm:nullclass}]
If we introduce the following vielbein:
\begin{align}
e_{\hat{t}} &= -\frac12 du - dv - \frac12\omega dx,\\
e_{\hat{x}} &= \sqrt{h} dx,\\
e_{\hat{y}} &= -\frac12 du + dv -\frac12\omega dx,
\end{align}
then the components of the vector $V$ defined above are given by:
\begin{align}
V_u &= 0 = -\frac12(\epsilon_1-\epsilon_2)^2,\\
V_v &= -1 = -(\epsilon_1+\epsilon_2)^2,\\
V_x &= 0= \frac12(\epsilon_1-\epsilon_2)\left[2\sqrt{h}(\epsilon_1+\epsilon_2)-\omega(\epsilon_1-\epsilon_2)\right].
\end{align}
From this we learn that $\epsilon_1=\epsilon_2$ and $\epsilon_1^2 = 1/4$. Thus, all such null SUSY solutions already break at least 1/2 of the supersymmetries.

One can now easily check that, with the above vielbein and imposing $\epsilon_1=\epsilon_2$, the supersymmetry variations (\ref{susy_var_chi}) identically vanish. 

Finally, again using the above vielbein and imposing $\epsilon:=\epsilon_1=\epsilon_2$, setting the supersymmetry variations (\ref{susy_var_psi}) to zero gives us:
\begin{align}
\partial_u \epsilon & =0,\\
\partial_x \epsilon  &=0,\\
\partial_v\epsilon +\frac14 Q_v \epsilon &=0.
\end{align}
From this we see that $\epsilon$ is only a function of $v$.

We can perform a gauge transformation of the metric, $g_{\mu\nu}\rightarrow g'_{\mu\nu}=g_{\mu\nu}+2\nabla_{(\mu}\xi_{\nu)}$, with parameter $\xi_{\mu} =(c(x,v), f(x,v), e(x,v) )$. Demanding that the form of the metric remains that of (\ref{eq:nullmetric}) gives us the condition that $c$ is a constant and further that:
\be e = -\frac{h \partial_v f}{\partial_v\omega} +c \omega.\ee
Further, the condition that the new metric has $\omega=0$, i.e.\ that $g'_{xv}=0$, is precisely equation (\ref{eq:nullfcondition}). So, if equation (\ref{eq:nullfcondition}) has a non-trivial solution for $f(x,v)$ and $c$, then this gauge transformation with parameter $\xi$ brings us to a metric with the form:
\be ds^2 = -2du dv + h_1(x,v)^2 dx^2,\ee
where $h_1$ is still in principle a function of $v$ \emph{and} $x$. The only non-trivial component of the Einstein equation is still the $(v,v)$ component and is now given by:
\be  \partial_v^2 h_1(x,v) +  h_1(x,v) P_v^A(v) P_v^A(v)=0.\ee
It is clear that the $x$-dependence of $h_1$ must be given by an arbitrary multiplicative function, i.e.\ $h_1(x,v) = t(x) h_0(v)$, where $t(x)$ is an arbitrary function of $x$. However, this function $t(x)$ can be absorbed in a redefinition of the coordinate $x$, so we can set $h_1=h_0(v)$ without loss of generality, finally obtaining the expression (\ref{eq:nullmetricnew}) for the metric with (\ref{eq:nullclassEOMnew}) as only equation of motion.
\end{proof}

\section{Simple Single Center Brane Solutions}\label{sec:simplebranes}
A relevant question that we want to investigate is to see what kind of branes are possible in the theory of 3D maximal supergravity. Thus, we are interested in static point particles in the timelike class above. The analysis in the section above shows us what the structure of $P_z$ must be for the brane to preserve supersymmetry, but this says nothing about the charge of the branes that can preserve SUSY\@. The ``charge'' of a point particle in 3D (which is an exotic (7-)brane in 10D \cite{deBoer:2010ud, deBoer:2012ma}) is its scalar monodromy around the brane; a discussion of the possible point particles in 3D is really a discussion of the possible scalar monodromies one can have. Ideally, we would like to find a (unambiguous) relation between $P_z$ and the monodromy of a (single) brane to be able to classify the possible single center monodromies according to supersymmetry preserved; however, as we will eventually realize in this section and the next, it appears that there is no simple relation between these two quantities, which unfortunately makes a classification of possible (supersymmetric) brane monodromies very difficult. This is similar to the discussion in \cite{Bergshoeff:2013sxa}, where it is shown using group theory arguments that the supersymmetry conditions for branes of codimension two or lower are degenerate, i.e.\ there are multiple distinct branes that can preserve exactly the same supersymmetries.

All of the explicit brane solutions considered here, as well as those of the next section, are all summarized in a more or less self-contained fashion in section \ref{sec:potformtaxonomy}. The reader only interested in the punchline of this section and the next can skip directly to there.

\subsection{Ansatz}
We will focus on a class of static spacetimes where the point particle sits at the origin and the metric is rotationally invariant; the metric is:
\be
ds^2 = -dt^2 + e^{U(r)} (dr^2 + r^2 d\varphi^2).
\ee
The usual complex coordinate is related to these by $z=r e^{i\varphi}$.

For the scalars, it will be convenient to work with the quantity $M=e^{\phi} e^{\phi^T}$ as defined around \eqref{eq:Mdef}; we are using the notation $x^T=-\theta(x)$. We will propose the following ansatz for $M$:
\be
\label{eq:Mansatz}
M = e^{\varphi\, X} m(r) e^{\varphi\, X^T}.
\ee
Here, $X$ is some algebra element in the algebra $\la{e}_{8(8)}$ and $m$ is a group element in $E_{8(8)}/SO(16)$, i.e.\ it satisfies $m^T=m$. We would like to emphasize that we have not proven that all supersymmetric single-center solutions must be of this type. Intuitively, the angular dependence in this ansatz corresponds to a geodesic in the coset $G/K$, and modifying this behavior would naively increase the energy. Still, it would be nice to prove or disprove that this ansatz captures all single-center supersymmetric solutions.

The scalar equation of motion for $M$ given by (\ref{eq:Mansatz}) is now:
\be
\label{eq:MansatzEOM}
r \partial_r (r\partial_r m\ m^{-1}) + [X, m X^T m^{-1}] = 0 .
\ee
This differential equation has many interesting properties, like the existence of an infinite set of conserved charges, which we will discuss in more detail in section \ref{sec:morebranes}.

It is also easy to see (in unitary gauge) that: \be \label{eq:MansatzP}
e^{\phi} P_{\mu} e^{-\phi} = \frac12 \partial_{\mu} M M^{-1}, \ee so
that $P_z$ is conjugate to $\partial_z M M^{-1}$. Also, one of the
Einstein equations of motion will imply the following constraint on $M$:
\be \label{eq:MansatzEOMconstraint} R_{zz} = \frac14\Tr(\partial_z M
M^{-1} \partial_z M M^{-1}) = 0.  \ee The Einstein equation of motion
for $R_{z\zb}$ determines the metric function $U$: \be
\label{eq:MansatzEOMU} r\partial_r (r\partial_r U) = -\frac14\tr\left[
(r\partial_r m m^{-1})^2 + (X+mX^Tm^{-1})^2\right].  \ee Of course, the
function $U$ that solves this equation must be well-behaved for the
solution to make physical sense; i.e.\ it must be real and not blowing
up except perhaps at the origin or at infinity in a way allowed by
string theory. In fact, we will not worry about the asymptotic behavior
of $r\rightarrow\infty$; single codimension-2 branes notoriously have
divergences at spatial infinity but our solutions should only be considered
to be valid approximations near an object and should be regulated in
some way to avoid divergences far away from the brane.

Since the scalars transform under a global gauge transformation as (\ref{V->GVH}), it is easy to see that $M$ transforms as:
\be M\rightarrow g M g^T.\ee
Therefore the scalar monodromy for the ansatz (\ref{eq:Mansatz}) can be identified as $g=e^{2\pi X}$. (In the full quantum theory, we would have to demand that the monodromy sits in the discrete $U$-duality group, $g\in G(\bbZ)$. This constraint on $g$, though clearly important, will not play any role in our subsequent analysis and we will therefore ignore it.) We will also call the algebra element $X$ the monodromy in what follows.

We wish to point out that the monodromy $X$ as defined here is not unambiguous. For a trivial example that illustrates this, consider $m=\mathbf{1}$. Then if $M = \mathbf{1}$, we can choose as ``monodromy'' $X$ any element in $\la{so}(16)$ and still keep $M=\mathbf{1}$. For generic $m$'s, the monodromy $X$ might have such (compact) ambiguities $X\sim X + k$ with $k$ in some subset of $\la{so}(16)$.

If we have a solution for a given monodromy $X$ (and $m$), then we can obtain any solution with a conjugate monodromy easily as follows. For a given conjugacy matrix $U\in G(\bbR)$, we take $X' = U X U^{-1}$ and $m'=UmU^T$; this implies $M \rightarrow U M U^T$. The scalar equation of motion (\ref{eq:MansatzEOM}) can easily be seen to be invariant under conjugation, so this procedure clearly produces a solution with the conjugate monodromy $e^{2\pi X'}=U e^{2\pi X} U^{-1}$. So we are once again interested in orbits; this time the orbits of the monodromies $X$, which will be $G(\bbR)$-orbits in $\la{g}_\bbR$.


\subsection{Review of $SL(2)$ Solutions}
First, we will review three basic building blocks that we will use in constructing our solutions later. These building blocks are solutions in $\la{sl}(2)$ (which means that $X\in \la{sl}(2)_\bbR$ and $m\in SL(2,\bbR)/SO(2)$); they are mostly already present in \cite{Bergshoeff:2006jj}.

Our convention for the matrix generators of $\la{sl}(2)_\bbR$ is:
\be H_{\la{sl}(2)} = \left(\begin{array}{cc} 1 & 0\\ 0 & -1\end{array} \right), \qquad X_{\la{sl}(2)} = \left(\begin{array}{cc} 0 & 1\\ 0 & 0\end{array} \right), \qquad Y_{\la{sl}(2)} = \left(\begin{array}{cc} 0 & 0\\ 1 & 0\end{array} \right).\ee

We wish to discuss all possible $\la{sl}(2)$ monodromies. Thankfully, up
to conjugacy, this only means considering three solutions which we call
$N$-, $A$-, and $K$-branes  (these are named after the factors in the
Iwasawa decomposition \eqref{eq:V=nak}).  Every element in $\la{sl}(2)$
is either nilpotent or semisimple. The nilpotent monodromies will all be
conjugate to the $N$-brane given below; note that there are two nilpotent $SL(2,\mathbb{R})$-orbits in $\la{sl}(2,\mathbb{R})$, of which representatives are the $N$-brane monodromy with either the plus or minus sign. The semisimple elements can be
divided into two classes: the compact elements (with eigenvalues $\pm i
\lambda$) and the non-compact elements (with eigenvalues $\pm
\lambda$). The former type of monodromy will be conjugate to the
$K$-brane given below, while the latter will be conjugate to the
$A$-brane given below.

In terms of known objects, the $N$-brane corresponds to the D7-brane (or anti-D7) brane solution; the $K$-brane and $A$-brane can be thought of as $(p,q,r)$ 7-branes with determinant (of the monodromy $X$) resp. larger and smaller than 0 \cite{Bergshoeff:2006jj}. The status of such $(p,q,r)$ 7-branes in string theory is somewhat unclear, although it is clear that our $A$-brane solution below is pathological and cannot correspond to any physical single-center brane solution.

Note that all solutions considered here have nilpotent $P_z$; this means that they can be embedded as supersymmetric solutions in our maximal $\la{e}_8$ theory. This is no accident; one can see that the equations of motion (\ref{eq:MansatzEOM}) together with the constraint (\ref{eq:MansatzEOMconstraint}) actually imply that $P_z$ is nilpotent; this is a special feature of such $\la{sl}(2)$ solutions.

The solutions we give here are a particular (simple) solution of the given type; a more complete family of solutions is derived for each type in section \ref{sec:warmupsl2}.

\subsubsection{$N$-brane ($D7$-brane)}\label{sec:Nbrane}
The $N$-brane, or $D7$-brane, has a nilpotent monodromy. The full solution is given by:
\begin{align}
 X &= X_{\la{sl}(2)} = \left(\begin{array}{cc} 0 & \pm 1\\ 0 & 0\end{array} \right),\\
 m &= \exp\left( \log(\log r)\left(\begin{array}{cc} 1 &0\\ 0 & -1\end{array} \right)\right) = \left(\begin{array}{cc} \log r &0\\ 0 & (\log r)^{-1}\end{array} \right).
\end{align}
It is easily checked that the equations of motion (\ref{eq:MansatzEOM})-(\ref{eq:MansatzEOMconstraint}) are satisfied, and moreover that $P_z$ is nilpotent (see section \ref{sec:warmupsl2}). We note that this solution is well-defined globally (modulo logarithmic divergences at infinity, as we are familiar with from the $D7$-brane) except at the origin $r=0$, where the brane is sitting. The plus and minus sign in the monodromy $X$ correspond to the D7-brane and the anti-D7 brane, which are not conjugate to each other; their respective $P_z$ elements are also in different conjugacy classes (i.e.\ the well-known fact that the SUSY preserved by a D7-branes is different than the SUSY preserved by an anti-D7).

\subsubsection{$K$-brane}\label{sec:Sbrane}
The $K$-brane has a semi-simple, compact monodromy. The solution is:
\begin{align}
 X & = \left(\begin{array}{cc} 0 & +\lambda\\ -\lambda& 0\end{array} \right),\\
 m &= \exp\left(2\tanh^{-1}(r^{2|\lambda|}) \left(\begin{array}{cc} 1 & 0\\ 0 & -1\end{array} \right)\right) = \left(\begin{array}{cc} \frac{1+r^{2|\lambda|}}{1-r^{2|\lambda|}} & 0\\ 0 & \frac{1-r^{2|\lambda|}}{1+r^{2|\lambda|}}\end{array} \right) .
\end{align}
Again, the equations of motion are satisfied. $P_z$ is nilpotent (see section \ref{sec:warmupsl2}). This solution is also well-defined globally; an important difference with the $N$-brane is that this solution remains finite and well-defined at the origin $r=0$, where the brane sits. Note that two $K$-branes with $\lambda=x$ and $\lambda=-x$ do not have conjugate monodromies; also, their respective corresponding $P_z$ elements will lie in different conjugacy classes (just like the D7 and anti-D7 branes).

\subsubsection{$A$-brane}\label{sec:Hbrane}
The $A$-brane has a semi-simple, non-compact monodromy. The solution is:
\begin{align}
 X & = \left(\begin{array}{cc} \lambda & 0\\ 0 & -\lambda\end{array} \right),\\
 m &= \left(\begin{array}{c@{\quad}c} \sec v & \tan v\\ \tan v & \sec v\end{array} \right), \qquad v\equiv c_1+2\lambda\log r,
\end{align}
where $c_1$ is an arbitrary constant. The equations of motion are satisfied; moreover $P_z$ is nilpotent (see section \ref{sec:warmupsl2}). This solution (even modulo issues at infinity) is by no means well-defined globally; if we get close to the origin, the solution oscillates wildly which is certainly unphysical; we are forced to excise an interval around $r=0$ from the spacetime to be able to make any sense of the solution. Note that two $A$-branes with $\lambda=x$ and $\lambda=-x$ have conjugate monodromies and equal $P_z$ elements (in contrast to the $K$-brane case).

One might wonder if there exist other possible $\la{sl}(2)$ solutions with the same monodromy but better behavior around $r=0$ (like the $N$- and $K$-brane have). Unfortunately, one can see from the equations of motion (\ref{eq:MansatzEOM}) and the constraint (\ref{eq:MansatzEOMconstraint}) that there are no other possibilities with this particular monodromy.


\subsection{Nilpotent Charges and Brane Representatives}\label{sec:nilpotentcharges}
We can easily construct branes with any nilpotent charge $X$ by embedding the $\la{sl}(2)$ $N$-brane solution described above into $\la{e}_{8(8)}$. Take any $\la{sl}(2)$ embedding in $\la{e}_{8(8)}$ which respects the Iwasawa decomposition, i.e.\ where the image of $\{H_{\la{sl}(2)},X_{\la{sl}(2)},Y_{\la{sl}(2)}\}$ is a (real) Cayley triple (see also section \ref{sec:lieconceptstriples}). The fact that the embedding respects the Iwasawa decomposition can also be stated as having the matrix-transpose operation $T$ in $\la{sl}(2)$ map to the $-\theta$ operation in $\la{e}_{8(8)}$. This then also assures us that $M\in G/H$. Embedding an $N$-brane in this way, the nilpotent monodromy $X$ of the resulting brane is then given by the image of $X_{\la{sl}(2)}$ under the embedding. This means that the resulting $P_z\bowtie X$ under the embedding\footnote{The relation $\bowtie$ is defined in (\ref{eq:defKSbij}).}; the (nilpotent) orbit of $X$ thus determines the supersymmetry of the solution.

A nilpotent element $X$ which is part of a Cayley triple $\{H,X,Y\}$ is explicitly tabulated by Djokovic \cite{Djokovic:reps} for every real nilpotent orbit in $\la{e}_{8(8)}$. The construction directly above shows us how to explicitly construct a brane solution with this monodromy $X$. To obtain any other nilpotent monodromy $X'$ in the orbit of $X$, we need only take the conjugate solution as described above. In this way we can, in principle, always find a solution given an arbitrary nilpotent monodromy $X'$.

\subsubsection{Brane Representatives}\label{sec:branereps}
While the paper \cite{Djokovic:reps} contains a list of explicit representatives for all of the nilpotent orbits under discussion, this does not give any intuition as to what types of brane intersections lie in which orbits. Thus, we would like to construct an explicit brane representative of each orbit that we can interpret easily in terms of well-known objects such as D-branes etc. We will search for a ``simple'' brane representative for each of the orbits in Table \ref{tab:orbits}.

To construct such brane representatives that we can interpret as regular (D-)brane intersections, we must first identify T-duality multiplets within the $\la{e}_8$ algebra. We will identify the 240 root vectors $E_{\alpha}$ (see (\ref{eq:e8roots}) in appendix \ref{app:twoconstructions} for our conventions of root vectors in $\la{e}_8$) as the 240 ``fundamental'' 1/2 BPS branes \cite{Bergshoeff:2013sxa,Bergshoeff:2012ex,Bergshoeff:2011se}. We will postulate that the eigenvalue of the root vector under the commutator with $n=2H_8-2$ is the power of the string tension $(g_s)^n$ that the mass of the brane is proportional to; it is easy to see that the centralizer of $n$ (not counting $H_8$ itself) is the expected T-duality algebra for a compactified $T^7$, namely $\la{so}(7,7)$. The value of $\alpha_8$ is $0, \pm 1/2, \pm 1$ so that $n\in\{0, -1, -2,-3,-4\}$. The value of $n$ thus determines which T-duality multiplet the brane sits in: $n=-1$ is the D-brane multiplet, $n=0$ is the F1/P multiplet, $n=-2$ is the KKM/NS5/$5^2_2$ multiplet, and $n=-3,n=-4$ are heavy, exotic brane multiplets (which we will not need in constructing representatives).

To further identify every brane within a particular T-duality multiplet,
we can postulate the following formula relating the mass of a brane
(represented by a particular root vector $\alpha$) to the string
coupling and the radii $R_i$ of the internal torus direction
$i$:\footnote{See \cite{Obers:1998fb} for similar mass formulas
in terms of compactification radii and roots.}
\be \label{eq:repmassformula} M \sim g_s^{2\alpha_8-2} \prod_{i=1}^7 R_i^{\alpha_i-\alpha_8+1}.\ee
This formula allows us to identify each of the 240 roots as a particular brane in type IIB compactified on $T^7$. We can take $\alpha_i\rightarrow -\alpha_i$ (for all $i\in\{1,\cdots,7\}$) in the formula to identify each root as a particular brane in type IIA on $T^7$, related to the IIB interpretation by 7 T-dualities along all the directions of the internal $T^7$. One can check that the mass formula (\ref{eq:repmassformula}) reproduces all of the correct number of each brane in every T-duality multiplet.

As an example of how we identify the branes within a T-duality multiplet, consider the D-brane multiplet ($n=-1$ so $\alpha_8=+1/2$) which is a \textbf{64} multiplet, containing (in IIB) 1 D7-brane, 21 D5-branes, 35 D3-branes, and 7 D1-branes. We note that the mass formula (\ref{eq:repmassformula}) shows us that the sign of the first seven positions of the root vector indicate whether a brane wraps this direction in the compact 7D space or not. For example:
\be D3(2,4,5) \leftrightarrow \frac12(-1,+1,-1,+1,+1,-1,-1|+1).\ee

To find the brane representatives listed in the fifth column of table \ref{tab:orbits}, we simply used trial and error (coupled with prior knowledge of supersymmetric brane intersections). We also checked explicitly that the 10D projectors of the brane representatives listed preserve the amount of SUSY mentioned. To check that a given brane representative does, in fact, lie in a particular orbit, we used the characterization of the orbit by means of its dimension $\dim$ and invariant $\inv$ as explained in appendix \ref{sec:appbranerep}\@.

To find an explicit solution for any given brane representative listed
there, we need only take the embedding of the $\la{sl}(2)$ generated by
$\{H,X,Y\}$ (with $X$ the relevant brane representative) into
$\la{e}_{8(8)}$ as described above. The $N$-brane then gives us the
standard 7-brane solution, reduced from 10D\@.

It is interesting to consider the representative of orbit $\mathbf{14}$
in Table \ref{tab:orbits}, since this is in a sense the ``maximal'' SUSY
orbit and representatives of all other orbits should be obtainable from
this representative by deleting branes (because the closure of orbit
$\mathbf{14}$ contains all other SUSY-preserving orbits as is evident
from fig.\ \ref{fig:hassediagram} in appendix \ref{sec:theorems})\@. We
wish to point out two interesting M-theory frames for this
representative by duality chains. The first is a collection of M5's and
one $P$:
\begin{align}
\left(\begin{array}{c}
 D4 (  3456 )\\
 D4 (12  56 )\\
 D4 (1234   )\\
 D0 \\
 KKM(34567,1)\\
 KKM(12567,3)\\
 KKM(12347,5)\\
 F1 (      7)\\
\end{array}\right)
 \xrightarrow{T_{1357}}
\left(\begin{array}{c}
 D4 (1  4 67)\\
 D4 ( 23  67)\\
 D4 ( 2 45 7)\\
 D4 (1 3 5 7)\\
 NS5(  34567)\\
 NS5(12  567)\\
 NS5(1234  7)\\
 P  (      7)\\
\end{array}\right)
 \xrightarrow{M_A} 
\left(\begin{array}{c}
 M5(1  4 67A)\\
 M5( 23  67A)\\
 M5( 2 45 7A)\\
 M5(1 3 5 7A)\\
 M5(  34567 )\\
 M5(12  567 )\\
 M5(1234  7 )\\
 P (      7 )\\
\end{array}\right)
 =
\left(\begin{array}{c}
 M5(1~~4~67A)\\
 M5(~23~~67A)\\
 M5(~2~45~7A)\\
 M5(1~3~5~7A)\\
 M5(~~3456~A)\\
 M5(12~~56~A)\\
 M5(1234~~~A)\\
 P(A)\\
\end{array}\right)
\label{eq:M5^7-P}
\end{align}
Note that the starting configuration is different from the
representative listed in Table \ref{tab:orbits} by a T-duality along the
$7$ direction and a relabelling of directions $1\leftrightarrow 2$ and
$5\leftrightarrow 6$. 
Also, $T_{1357}$  means $T$-duality along 1357 directions while $M_A$ means
M-theory lift with $x_A$ being the M-circle direction.
The brane configuration in the last frame looks
like a 1/16-BPS generalization of the Maldacena-Strominger-Witten system
\cite{Maldacena:1997de}, which is a 1/8-BPS system of three intersecting M5-branes with
momentum along the intersection.  In 
\eqref{eq:M5^7-P},
any pair of M5-brane stacks shares 3
directions including $A$ which is common to all 7 M5-branes stacks. 
The 11D metric  for this configuration is, by the harmonic function rule,
\begin{align}
ds_{11}^2&=
 (1\cdots 7)^{-1/3} [-dt^2+dx_A^2+H_0(dt+dx_A)^2]
 + (1\cdots 7)^{2/3} (dx_8^2+ dx_9^2)
\notag\\
 &\qquad
 + (1467)^{-1/3} (235)^{2/3} dx_1^2
 + (2367)^{-1/3} (145)^{2/3} dx_2^2
 + (2457)^{-1/3} (136)^{2/3} dx_3^2
\notag\\
 &\qquad
 + (1357)^{-1/3} (246)^{2/3} dx_4^2
 + (3456)^{-1/3} (127)^{2/3} dx_5^2
\notag\\
 &\qquad
 + (1256)^{-1/3} (347)^{2/3} dx_6^2
 + (1234)^{-1/3} (567)^{2/3} dx_7^2,
\end{align}
where e.g.\ ``$1\cdots 7$'' is a shorthand notation for $H_1\cdots
H_7$. Also, $H_I=Q_I \log(1/r)$, $r=\sqrt{x_8^2+x_9^2}$, where
$I=0,\dots,7$.  Here we assume that $Q_I>0$. The 3-form potential is
\begin{align}
 A_3&=
 -K_1 dx^{235}-K_2 dx^{145} -K_3 dx^{136}\notag\\
 &\qquad\qquad
 +K_4 dx^{246}+K_5 dx^{127}+K_6 dx^{347}+K_7 dx^{567}
\end{align}
where $dx^{235}=dx^2\wedge dx^3 \wedge dx^5$ etc.  Also, $K_I=Q_I
\theta$ where $\theta=\arctan(x_9/x_8)$.  The signs in $A_3$ are chosen so that
this represents a supersymmetric configuration.  Unlike the original MSW
system, this 1/16-BPS system has vanishing horizon area.

A second interesting M-theory frame involves M2-branes and KK-monopoles;
the duality chain is (starting from the same IIA starting frame as
above):
\begin{align}
\xrightarrow{T_{123456}}&
\left(\begin{array}{c}
 D2 (12     )\\
 D2 (  34   )\\
 D2 (    56 )\\
 D6 (123456)\\
 KKM(34567,2)\\
 KKM(12567,4)\\
 KKM(12347,6)\\
 F1 (      7)\\
\end{array}\right)
 \xrightarrow{M_A}
\left(\begin{array}{c}
 M2 (12     )\\
 M2 (  34   )\\
 M2 (    56 )\\
 KKM(123456,A)\\
 KKM(34567A,2)\\
 KKM(12567A,4)\\
 KKM(12347A,6)\\
 M2 (      7A)\\
\end{array}\right)
=
\left(\begin{array}{c}
 M2 (12~~~~~~)\\
 M2 (~~34~~~~)\\
 M2 (~~~~56~~)\\
 M2 (~~~~~~7A)\\
 KKM(~~34567A,2)\\
 KKM(12~~567A,4)\\
 KKM(1234~~7A,6)\\
 KKM(123456~~,A)\\
\end{array}\right)
\end{align}
The brane configuration is appealingly symmetric in this last frame.


\subsection{Semi-simple Charges}\label{sec:semisimplecharges}
Every semi-simple element in $\la{e}_{8(8)}$ is conjugate to an element of the form $X'=\sum_i c_i H_{\alpha_i} + \sum_j  c_j K_{\alpha_j}$, where $i$ and $j$ take values in two disjoint subsets of $\{1,\cdots,8\}$, for a specific collection of 8 roots $\{\alpha_i\}$, where $K_{\alpha_j}=E_{\alpha_j}-E_{-\alpha_j}$ and $H_{\alpha_i}=\alpha_i\cdot \vec{H}$ is the neutral element in the triple $\{H_{\alpha_i},E_{\alpha_i},E_{-\alpha_i}\}$, so that $K_{\alpha_i}\in \la{so}(16)$ while $H_{\alpha_j}\in \la{e}_{8(8)}\ominus\la{so}(16)$ (lemma \ref{lemma:sugiuracartanconjugacy}; the explicit roots $\alpha_i$ are also given).

To find an explicit solution for elements of the form $X'$ is
straightforward: for each $H_{\alpha_i}$ appearing in the sum, we
introduce an $A$-brane that embeds its $\la{sl}(2)$ elements
$\{H_{\la{sl}(2)},X_{\la{sl}(2)},\ Y_{\la{sl}(2)}\}$ into
$\{H_{\alpha_i},E_{\alpha_i}, E_{-\alpha_i}\}$; for each $K_{\alpha_j}$
we use a $K$-brane that embeds its $\la{sl}(2)$ elements
$\{H_{\la{sl}(2)},X_{\la{sl}(2)},Y_{\la{sl}(2)}\}$ into
$\{H_{\alpha_j},E_{\alpha_j}, E_{-\alpha_j}\}$. For each $K$- and $A$-brane, the corresponding value of $\lambda$ in the $\la{sl}(2)$ monodromy is determined by the eigenvalues of $X'$. The roots $\alpha_i$
are chosen such that $\alpha_i\pm \alpha_j$ is not a root, so all of the
$\la{sl}(2)$-embeddings of the $K$/$A$-branes will mutually
commute. Thus, for the final solution, we can just paste the different
$K$- and $A$-branes together as follows:
\be
X_{mon} = X', \qquad m = m_{\alpha_1} \times m_{\alpha_2}\times \cdots.
\ee
The fact that all of the $\la{sl}(2)$-embeddings mutually commute assures us that the corresponding equations of motion for this configuration will factorize into the equations of motion for the different $K$/$A$-branes and will thus be automatically satisfied.

We also notice that the solution obtained in this way will have $P_z\bowtie\sum_j x^j E_{\alpha_j}$ for some coefficients $x^j$ (which are spacetime-dependent in general). What these coefficients are will determine which nilpotent orbit $P_z$ lies in and thus will determine the supersymmetry of the solution; i.e.\ the relative orientation of the SUSY preserved by the individual branes that we have pasted together will determine if the full solution preserves SUSY\@.


\subsection{Other Charges --- $\la{sl}(2)^n$}\label{sec:othercharges}
In general, an element in a Lie algebra has a unique decomposition into nilpotent and semisimple parts:
\be X = X_S + X_N, \qquad [X_N,X_S]=0.\ee
In $\la{sl}(2)$, no two (distinct, non-zero) elements satisfy $[X,Y]=0$, so there are only strictly nilpotent or strictly semisimple elements in $\la{sl}(2)$, as we have mentioned before.

One might be tempted to take the brane solutions for $X_N$ and $X_S$ as detailed above and try  to ``paste'' them together somehow. However, generically this will not be possible as e.g.\ we would have $[m_{X_N},m_{X_S}]\neq 0$ even though $[X_N,X_S]=0$. Only for very specific $X=X_S+X_N$ can we construct solutions by ``pasting''; for example if we have that the $\la{sl}(2)_N$ algebra commutes with the $(\la{sl}(2)^n)_S$ algebra, where $\la{sl}(2)_N$ is the embedding of $\la{sl}(2)$ that the nilpotent $N$-brane solution generates, and similarly $(\la{sl}(2)^n)_S$ is the embedding of the $\la{sl}(2)^n$ algebra that the semisimple brane solution generates. This is a very restricted class of solutions.

We can give one simple example, in $\la{sl}(2)\oplus \la{sl}(2)\subset \la{sl}(4)$. We can take e.g.:
\be X = \left(\begin{array}{cc} X_N & 0\\ 0 & X_S\end{array}\right) =  \left(\begin{array}{cccc} 0 & 1 &0&0\\ 0 & 0&0&0\\0 & 0& 0&1\\ 0&0&-1&0 \end{array} \right), \qquad m = \left(\begin{array}{cc} m_{N} & 0\\ 0 & m_S\end{array}\right),\ee
so that the solution is an obvious ``pasting'' together of one $N$- and $K$-brane. The supersymmetry of this solution depends on how we embed this $\la{sl}(2)\oplus \la{sl}(2)$ algebra into $\la{e}_{8(8)}$.

\section{More $\la{sl}(n)$ Brane Solutions}\label{sec:morebranes}
In this section we will continue our search for single-centered brane solutions with the ansatz (\ref{eq:Mansatz}), searching outside of $\la{sl}(2)^n$ subalgebras of the previous section by considering a number of $\la{sl}(n)$ monodromies. The explicit solutions of section \ref{sec:explicitex} will all live in $\la{sl}(3)$, but the general analysis of section \ref{sec:potform} will be valid for solutions living in any $\la{sl}(n)$ subalgebra of $\la{e}_8$. Note that when we make statements about powers of algebra elements, such as $P_z^2=0$ in section \ref{sec:condsusy}, unless explicitly stated otherwise we are considering these statements in the fundamental representation of the $\la{sl}(n)$ subalgebra, which is the representation of traceless $n\times n$ matrices. We will also use the terms ``group/algebra element'' and ``matrix'' interchangeably.

 To facilitate the detailed analysis of these $\la{sl}(n)$ solutions, we will first introduce a few new concepts, expanding the analysis of the previous section of the ansatz (\ref{eq:Mansatz}).

All of the explicit brane solutions considered here, as well as those of the previous section, are all summarized in a more or less self-contained fashion in section \ref{sec:potformtaxonomy}. The reader only interested in the punchline of this section and the last can skip directly to there.

\subsection{Potential Formalism}\label{sec:potform}
Using the ansatz (\ref{eq:Mansatz}), the scalar part of the Lagrangian (\ref{S_scalar_ito_M}) can be rewritten as:
\begin{align}
 L_{scal} &= -\frac14 \tr\left(\partial_u M^{-1} \partial_u M + \partial_{\theta} M^{-1}\partial_{\theta}M\right)\\
& = \tr\left(\frac14 m^{-1}\dot{m} m^{-1} \dot{m} + \frac12 X^T m^{-1} X m + \frac12 X^2\right)\\
 &= K - W,
\end{align}
where we have introduced a new radial coordinate $r=e^{-u}$, and $\dot{\,}=\partial_u$. We can interpret this Lagrangian as describing a
particle evolving in time $u$ with coordinates parametrized by $m(u)$, with kinetic energy $K$ and potential energy $W$ given by:
\begin{align}
 K &= \frac14 \tr\left(m^{-1}\dot{m} m^{-1} \dot{m}\right),\\
W &= -\frac12\left( X^T m^{-1} X m + X^2\right) = -\frac14 \tr\left( \left[ m^{-1}(Xm+mX^T)\right]^2\right) \leq 0,
\end{align}
where we have pointed out that the potential energy is always
negative. We will call the scalar fields that parametrize $m$
generically by $\phi$ (the ``coordinates'' of the particle), so that
when we talk about finding a critical point of $W$, this means finding a
solution $\phi_0(r)$ to the equations $\partial_{\phi} W=0$.

\subsubsection{Constraints and Conserved Charges}
The Einstein equation (\ref{eq:MansatzEOMconstraint}) gives us two constraints (since it is a complex equation). The first is an ``energy constraint'':
\be \label{eq:potformHconstraint} H \equiv K + W = 0,\ee
where we have defined the Hamiltonian $H$ for this system. The second constraint is a ``momentum constraint'':
\be \label{eq:potformcX} c_X \equiv \tr(X \dot{m} m^{-1}) = 0,\ee
where we have defined $c_X$.

We can find additional conserved quantities. Define, for an algebra element $Y$, the quantity:
\be \label{eq:potformcY} c_Y \equiv \tr(Y \dot{m}m^{-1}),\ee
then:
\be \dot{c}_Y = \tr([X,Y] m X^T m^{-1}),\ee
so that $c_Y$ is a conserved quantity if $[X,Y]=0$. This corresponds to the symmetry of the potential under $\delta m = Y m + m Y^T$. Using $c_Y$, we should be able to construct as many linearly independent conserved charges as the dimension of the centralizer of $X$, which (in $\la{e}_8$) is at least 8.

We can show that the orbit of $P_z$ remains the same everywhere for a given solution with our ansatz (\ref{eq:Mansatz}). We can define the quantity:
\be \hat{P}_z = \omega^{-1} e^{-\varphi X} r \partial_z M M^{-1} e^{\varphi X} \omega,\ee
where we have chosen a matrix $\omega$ to be real and symmetric and satisfy\footnote{This is always possible because $m$ is of the form $m=A^{\dagger}A$ so that $m$ is positive semi-definite.} $m = \omega^2$. Clearly, $P_z$ is conjugate to $\hat{P}_z$ (see (\ref{eq:MansatzP})) so they are in the same orbit. Now, using the equations of motion, we can find the following expression:
\begin{align}
r\partial_r \hat{P}_z &= [ \hat{P}_z, -\omega^{-1}\dot{\omega} - i (\omega X^T \omega^{-1})]\nn
&= [ \hat{P}_z, -\omega^{-1}\dot{\omega} - i (\omega X^T \omega^{-1}) - \hat{P}_z/2]\nn
&= [\hat{P}_z,\frac12\left(-\omega^{-1}\dot{\omega}+\dot{\omega}\omega^{-1}-i\omega X^T \omega^{-1}+i\omega^{-1}X\omega\right)] \equiv [\hat{P}_z, \hat{k}],
\end{align}
where $\hat{k}$ clearly satisfies $\hat{k}^T = -\hat{k}$ and thus is an
element of the compact Lie subalgebra $\la{k}$. We can conclude that the
orbit of $P_z$ --- which is the same as the orbit of $\hat{P}_z$ --- is
fixed once we have fixed a solution. It is not possible to ``move in and
out of orbits'' within one solution (excluding possible singular points in the spacetime where e.g.\ $P_z$ blows up). Since the orbit of $P_z$ determines
the amount of SUSY preserved, this means that an analysis of
supersymmetry for a solution with the ansatz (\ref{eq:Mansatz}) can be
performed only at one spacetime point (or rather, circle of constant $r$) and does not need to involve
knowledge of the function $P_z$ over the entire spacetime. (This is also as expected from the discussion at the end of section \ref{sec:maxsugra}.) We note also
that, since the orbit of $P_z$ is a constant in spacetime, this implies
that the quantities $\tr(P_z^k)$ are conserved charges for every $k$ (these quantities will all vanish if $P_z$ is nilpotent, because
then all eigenvalues of $P_z$ are zero). Note that this derivation is, in fact, independent of the
representation $R$ that we use when calculating $P_z^k$, so that indeed
all of the charges:
\be q_{R,k} = \tr_R P_z^k,\ee
should be conserved, for any integer $k$ and any representation $R$.


\subsubsection{A Class of Solutions with $W=0$}
Let us investigate the behaviour of solutions for $r\rightarrow 0$, i.e.\ $u\rightarrow \infty$. One natural class of solutions would have $K\rightarrow 0$ in this limit, which implies via the energy constraint (\ref{eq:potformHconstraint}) that the configuration should asymptote to a point where the potential vanishes, $W=0$. Moreover, since $K$ asymptotes to a constant, the point at which $W=0$ must also be a critical point of $W$, i.e.\ $\partial_{\phi}W=0$. 

Let us investigate first what $W=0$ implies. The potential can be rewritten as:
\be W 
= \frac14 \tr\left(  \left[ m^{-1/2} X m^{1/2} + m^{1/2} X^T m^{-1/2}\right]^2 \right) = -\frac14 \tr( Y^2) \leq 0.\ee
Now the matrix $Y$ as defined above is symmetric, so that $\tr(Y^2) = 0$ if and only if $Y=0$. This means $W=0$ if and only if:
\be \label{eq:potformW0cond} m^{-1/2} X m^{1/2} + m^{1/2} X^T m^{-1/2} = 0.\ee

Now we add the assumption that $m$ is finite. Then we can rewrite (\ref{eq:potformW0cond}) as:
\be \label{eq:potformW0condmfinite} X m + m X^T = 0.\ee
We can multiply this equation by $g$ from the left and by $g^T$ from the right, choosing $g$ such that $g m g^T=1$ (this is possible as $m$ is symmetric). Then we have:
\be \left(g X g^{-1}\right) + \left(g X g^{-1}\right)^T = 0,\ee
so that $X$ is $G$-conjugate to an antisymmetric element, $g X g^{-1} = X'\in K$. We note that since all elements of $K$ are semi-simple, $X$ is semi-simple as well. So we conclude that $X$ is conjugate to some semi-simple compact element in $K$. We also note that (\ref{eq:potformW0condmfinite}) is satisfied for $m=g^{-1} (g^{-1})^T$; so there certainly exists a solution where $W=0$ at the origin. A (conjugation of a) combination of $K$-brane solutions (as explained in section \ref{sec:semisimplecharges}) is an example of such a configuration where $W=0$ at the origin and $X$ is (conjugate to) a compact semi-simple element.


We can conclude that, if $W\rightarrow 0$ at a point where $m$ remains finite, then $X$ is (conjugate to) a compact semi-simple element. It is certainly possible to have $W\rightarrow 0$ for other monodromies $X$, but then we must have that components of $m$ diverge. This is the case in e.g.\ the $N$-brane solution, where $X$ is nilpotent, $W\rightarrow 0$ for $r\rightarrow 0$, but $m$ diverges logarithmically in this limit.

\subsubsection{Necessary Condition for Supersymmetry}\label{sec:condsusy}
Up until now, we have not considered the demand of supersymmetry. Here we will investigate what demanding supersymmetry restricts of $\la{sl}(n)$ solutions with the ansatz (\ref{eq:Mansatz}).

We know that $P_z$ must be nilpotent in order for there to be any
supersymmetry (see result \ref{prop:Pnilpotent} in section
\ref{sec:statementresults}), but nilpotency alone is not enough to
ensure supersymmetry; there are only particular nilpotent orbits which
preserve supersymmetry (see result \ref{prop:Porbits} in section
\ref{sec:statementresults} and table \ref{tab:orbits}). These orbits, as
noted in section \ref{sec:statementresults}, all have a structure of the
kind $\la{sl}(2)^n$. These orbits all have representatives of the form\footnote{Recall the definition of $\bowtie$ in (\ref{eq:defKSbij}).}
$P_z \bowtie E_{\alpha_1} + E_{\alpha_2} + \cdots + E_{\alpha_n}$ ($n\leq 8$)
where all the roots involved are such that $\alpha_i\pm \alpha_j$ is not
a root. This means that we can embed $P_z$ in an $\la{sl}(2)^n$
subalgebra of $\la{e}_8$, but in particular also implies that $P_z^2 =
0$ in the fundamental representation of an $\la{sl}(k)$ subalgebra of
$\la{e}_8$ that contains all $n$ of these $\la{sl}(2)$ subalgebras. Let
us denote this representation by $f$ and $P_z$ in this representation by
$P_z^f$ (note that $f$ is different from the fundamental representation
of $\la{e}_8$).  Since $(P_z^f)^2=0$ is a statement that is invariant
under conjugation (in this $\la{sl}(k)$), so we can conclude that
$(P_z^f)^2 = 0$ is a requisite for supersymmetry. This is a much
stronger demand on $P_z$ than nilpotency, because nilpotency just
demands that there be a finite $m$ such that $(P_z^f)^m = 0$.  Note also
that $(P_z^f)^2=0$ is a necessary but not sufficient
condition. Different embeddings of $\la{sl}(k)$ into $\la{e}_8$ could
give different supersymmetry preserved. 
For example, we saw in section \ref{sec:nilpotentcharges} that the
same $\la{sl}(2)$ solution (which
satisfies $(P_z^f)^2=0$) can be embedded into $\la{e}_{8(8)}$ in different ways to
give any possible nilpotent element $P_z$ with varied amount (including
zero) of supersymmetry; so, it is
important for supersymmetry which nilpotent orbit the embedding of $P_z$
falls into.

It is easily shown that $P_z$ is conjugate to:
\be P_z \sim \omega^{-1}\dot{\omega}+\dot{\omega}\omega^{-1} + i(\omega^{-1}X\omega + \omega X^T\omega^{-1}),\ee
where we have again selected $\omega$ such that $m=\omega^2$. If we set:
\be \mathcal{A} \equiv \omega^{-1}\dot{\omega}+\dot{\omega}\omega^{-1}, \qquad \mathcal{B} \equiv Y + Y^T, \qquad Y = \omega^{-1}X\omega,\ee
then we see that the demand that\footnote{We will leave out the superscript $f$ in the rest of this section, but all matrices should always be considered in the fundamental representation of $\la{sl}(k)$.} $P_z^2 = 0$ can be translated into the (equivalent by conjugation) condition that $(\mathcal{A}-i\mathcal{B})^2=0$, or:
\be \mathcal{A}^2 = \mathcal{B}^2, \qquad \{\mathcal{A},\mathcal{B}\} = 0,\ee
where $\{\cdot,\cdot\}$ is the anticommutator. After a rescaling, this means $\mathcal{A},\mathcal{B}$ are a representation of the Euclidean two-dimensional Clifford algebra. Since both matrices are also symmetric, that means we can conjugate $P_z$ so that $\mathcal{A}$ and $\mathcal{B}$ take the following form:
\be \mathcal{A} = \diag(b_1, b_2, \cdots, b_n)\otimes \sigma_3 + \mathbf{0}_m, \qquad \mathcal{B}=\diag(b_1, b_2, \cdots, b_n)\otimes \sigma_1+ \mathbf{0}_m, \ee
where $b_i\neq 0$, and $\mathbf{0}_m$ denotes a possibility of having an additional $m$ eigenvalues of 0. The most important thing to realize is that the non-zero eigenvalues of $\mathcal{A},\mathcal{B}$ thus must come in pairs: if $\lambda$ is an eigenvalue of either then so is $-\lambda$, and with the same multiplicity. For $Y$, we have:
\be Y = \frac12 \diag(b_1, b_2, \cdots, b_n)\otimes \sigma_1 + \mathbf{0}_m + A, \qquad A^T = -A.\ee

By doing a global $G$ transformation $m\rightarrow g m g^T$, we can assume that $m(u_0)=1$ at any given point $u=u_0$ (note that this also induces $X\rightarrow g X g^{-1})$. If $m(u_0)=1$, then the expression at $u=u_0$ become:
\be \label{eq:potformSUSYcond} \mathcal{A} = \dot{m}(u_0)=\diag(b_1, b_2, \cdots, b_n)\otimes \sigma_3+ \mathbf{0}_m, \qquad \mathcal{B} = X + X^T=\diag(b_1, b_2, \cdots, b_n)\otimes \sigma_1+ \mathbf{0}_m.\ee
Given this data at a point $u=u_0$, the equation of motion allows us in principle to find the entire solution for $m(u)$.

\subsection{Various Explicit Examples}\label{sec:explicitex}
We will now consider a variety of examples in an $\la{sl}(3)$ subalgebra
of $\la{e}_8$. First we will ``warm up'' by rederiving the $N$-, $K$-,
$A$-brane solutions for $\la{sl}(2)$ using the potential and the
energy/momentum constraints.

Our convention for the matrix generators of $\la{sl}(3)$ in the fundamental representation is the following:
\begin{align}
\begin{split}
  H_1 &= \left(\begin{array}{ccc} 1 & 0 & 0\\ 0&0&0\\ 0&0&-1\end{array}\right), \qquad H_2 = \frac{1}{\sqrt{3}}\left(\begin{array}{ccc} 1 & 0 & 0\\ 0&-2&0\\ 0&0&1\end{array}\right),\\
 X_{12} &= \left(\begin{array}{ccc} 0 & 1 & 0\\ 0&0&0\\ 0&0&0\end{array}\right), \qquad X_{13}=\left(\begin{array}{ccc} 0 & 0 &1\\ 0&0&0\\ 0&0&0\end{array}\right), \qquad X_{23} = \left(\begin{array}{ccc} 0 & 0 & 0\\ 0&0&1\\ 0&0&0\end{array}\right),\\
 Y_{12} &= X_{12}^T, \qquad Y_{13} = X_{13}^T, \qquad Y_{23} = X_{23}^T.
\end{split}
\end{align}
The scalar matrix will be parametrized by:
\be m(u) = v(u) v(u)^T, \qquad v(u) = e^{a_1 X_{23}}e^{a_2 X_{13}}e^{a_3 X_{12}}e^{\frac12(\phi_1H_1+\phi_2H_2)},\ee
so that the functions of $u$ are $a_1,a_2,a_3,\phi_1,\phi_2$.

\subsubsection{Warm Up: $\la{sl}(2)$}\label{sec:warmupsl2}
The algebra $\la{sl}(2)$ is small enough to allow a more or less comprehensive study of possible monodromies. Of course, we have studied a solution for every possibility of monodromy---the $N$-brane for nilpotent $X$, the $K$-brane for semisimple, compact $X$, and the $A$-brane for semisimple, noncompact $X$. We will see here how these three solutions have a natural place within this potential formalism. We will only consider the $K$- and $A$-branes with $\lambda=1$ here for simplicity.

We take as ansatz:
\begin{align} m(u) &= e^{a(u) X_{\la{sl}(2)}} e^{\phi(u) H_{\la{sl}(2)}} e^{a(u) Y_{\la{sl}(2)}} = \left(\begin{array}{cc} e^{\phi} + a^2 e^{-\phi} & a e^{-\phi}\\ a e^{-\phi} & e^{-\phi}\end{array}\right),\\
X &= e X_{\la{sl}(2)} + f Y_{\la{sl}(2)} + h H_{\la{sl}(2)} = \left(\begin{array}{cc} h & e\\ f & -h\end{array}\right).
\end{align}
In giving $P_z^f$, we will use (see also (\ref{eq:defPQ})):
\begin{align}
 v(u) &=  e^{a(u) X_{\la{sl}(2)}} e^{\frac12\phi(u) H_{\la{sl}(2)}}, & (m(u)&=v(u)v(u)^T),\\
 V(u,\varphi) &= e^{\varphi X} v(u), & P_z &= \frac{1-\theta}{2} V^{-1}\partial_z V.
\end{align}
Our kinetic and potential energies become:
\begin{align}
 K &= \frac12\left( \dot{a}^2 e^{-2\phi} + \dot{\phi}^2\right),\\
 W &= -e\, f - \frac12 e^{2\phi}f^2 -2h^2+2f\, h\, a-f^2\, a^2 - \frac12 e^{-2\phi}(e+2h\, a - f\, a^2)^2.
\end{align}
The momentum constraint is:
\be c_X = \dot{a}(-a^2 f + e + 2a h) e^{-2\phi} - 2 a f \dot{\phi} + \dot{a} f + 2h\dot{\phi} = 0.\ee
For this $\la{sl}(2)$ system, since we have only two variables $a,\phi$, we see that the two constraints (energy and momentum) in principle completely determine the solution in function of the constants $e,f,h$. Let us study the three cases which lead to the $N,A,K$-branes:
\begin{enumerate}
 \item \emph{$N$-brane, $(e,f,h)=(1,0,0)$:}\\
 The potential energy and momentum become:
\begin{align}
 W& = -\frac12 e^{-2\phi},\\
c_X &=  \dot{a} e^{-2\phi}.
\end{align}
The energy constraint $H=K+W=0$ and the momentum constraint $c_X=0$ are solved by:
\begin{align}
 a(u) &= a_0,\\
 \phi(u) &= \log(\pm (u-u_0)).
\end{align}
We note that, for $u\rightarrow\infty$, we have $W\rightarrow 0$ as $\phi$ (and thus $m$) diverges. The $N$-brane of section \ref{sec:Nbrane} (with the plus sign in the monodromy) is this solution with the specific choice $a_0=0, u_0=0$ and the minus sign in the expression for $u$. $P_z^f$, for the plus sign in $\phi(u)$ (which is the only physically sensible solution) is given by:
\be P_z^f=  \frac{e^u}{2(u-u_0)} \left( \begin{array}{cc} 1 & i\\ i & -1 \end{array}\right),\ee
which clearly satisfies $(P_z^f)^2=0$. For $e=-1$, we would get the same
       expression for $P_z$ but with different signs for the $(12)$ and
       $(21)$ components, which means $P_z$ for $e=-1$ is in a different
       orbit from $P_z$ for $e=+1$.

\item \emph{$K$-brane, $(e,f,h)=(1,-1,0)$:}\\
The potential energy and momentum are given by:
\begin{align}
 W &= 1- a^2 - \frac12 e^{2\phi} - \frac12 e^{-2\phi}(1+a^2)^2,\\
c_X &= \dot{a}(a^2+1)e^{-2\phi} +2a\dot{\phi} - \dot{a}.
\end{align}
The general solution to the constraint equations is given by:
\begin{align}
 \phi &= \log \frac{(c^2+4)e^{4(u-u_0)}-1}{4e^{4(u-u_0)} + (c e^{2(u-u_0)}-1)^2},\\
 a &= \pm \frac{4e^{2(u-u_0)}}{4e^{4(u-u_0)} + (c e^{2(u-u_0)}-1)^2}.
\end{align}
For $u\rightarrow\infty$, we have $W\rightarrow 0$ and $m\rightarrow 1$. The $K$-brane of section \ref{sec:Sbrane} is this solution with $c=0,u_0=(1/2) \log 2$ and choosing the minus sign for $a$. We have:
\be P_z^f = 2e^{u} \left( \frac{\mp 2i+c}{e^{-2(u-u_0)}-(\mp2i+c)} - \frac{4+c^2}{e^{-4(u-u_0)}-(4+c^2)} \right)   \left( \begin{array}{cc} 1 &  - i\\ -i & -1 \end{array}\right), \ee
which satisfies $(P_z^f)^2 = 0$. For $e=-1,f=+1$, we would get the same expression for $P_z$ but with different signs for the $(12)$ and $(21)$ components (just like for the $N$-brane), which signifies that the two $P_z$'s lie in different nilpotent orbits.

\item \emph{$A$-brane, $(e,f,h)=(0,0,1)$:}\\
The potential and momentum are:
\begin{align}
 W &= -2 -2 a^2 e^{-2\phi}<0,\\
 c_X &= 2\dot{a}a e^{-2\phi} + 2\dot{\phi}.
\end{align}
The general solutions are given by:
\begin{align}
 \phi &= \log[ 2c \cos(2(u-u_0))],\\
 a &= \pm 2c \sin (2(u-u_0)).
\end{align}
Note that for $u\rightarrow\infty$, the solution remains oscillating. It is also clearly not possible to ever reach $W=0$. The $A$-brane of section \ref{sec:Hbrane} is this solution with $c=1/2, u_0= c_1/2$, and choosing the plus sign for $a$. We get:
\be P_z^f = -e^u (\tan2(u - u_0) + i) \left( \begin{array}{cc} 1 &  \pm i\\ \pm i & -1 \end{array}\right), \ee
which again satisfies $(P_z^f)^2=0$. Note that for the same monodromy, we are able to construct an $A$-brane with $P_z$ in both possible conjugacy classes (as opposed to both the $N$- and $K$-brane).
\end{enumerate}

\subsubsection{$\mathbf{sl}(3)$ Example 1: simple nilpotent}\label{sec:potformsl3X13}
Take as monodromy:
\be X = X_{13}=\left(\begin{array}{ccc} 0 & 0 & 1\\ 0&0&0\\ 0&0&0\end{array}\right).\ee
Clearly there is a solution in an $\la{sl}(2)$ subalgebra that has $(P_z^f)^2=0$ and corresponds to an $N$-brane where we embed $X_{\la{sl}(2)}$ onto $X=X_{13}$. However, we are interested in finding a different solution with this same monodromy; a solution that is ``truly $\la{sl}(3)$'' in the sense that it can not be given by an $\la{sl}(2)$ solution embedded in $\la{sl}(3)$.

The potential is given by:
\be W = -\frac12 e^{-2\phi_1}.\ee
The momentum constraint is:
\be \dot{a}_2 - a_3\dot{a}_1 = 0.\ee
One can use (\ref{eq:potformcY}) to obtain the following three conserved charges:
\begin{align}
\begin{split}
 c_{11} &= -e^{-\phi_1+\sqrt{3}\phi_2} a_1\dot{a}_1 + e^{-\phi_1-\sqrt{3}\phi_2}a_3\dot{a}_3 +\frac{2}{\sqrt{3}} \dot{\phi}_2,\\
 c_{12} &= e^{-\phi_1-\sqrt{3}\phi_2}\dot{a}_3,\qquad\qquad
 c_{23} = e^{-\phi_1+\sqrt{\phi_2}}\dot{a}_1,
\end{split}
\end{align}
where we have used $Y=X_{11}+X_{33},X_{12},X_{23}$, respectively, and
already used $c_X=0$. In principle, the five conserved charges $H,c_X,
c_{11},c_{12},c_{23}$ are enough to determine the variables
$\phi_1,\phi_2,a_1,a_2,a_3$. The $N$-brane embedding where $X_{\la{sl}(2)}$
is embedded onto $X$ has $\phi_1=\log u;$ $\phi_2,a_1,a_2,a_3=const.$; $c_{11}=c_{12}=c_{23}=0$; this is the only possible solution with
$(P_z^f)^2=0$.

To find a ``truly $\la{sl}(3)$'' solution, we take the following ansatz:
\begin{align}
\begin{aligned}
 c_{11} &= 0, & \qquad\qquad c_{12}&=c_{23}\equiv c\neq 0,\\
 a_1(u) &= a_3(u)\equiv a(u), & a_2(u)&= \frac12 a(u)^2 + a_{20},\\
 \phi_1(u) &= \log \frac{\dot{a}(u)}{c}, & \phi_2(u) &= 0.
\end{aligned}
\end{align}
The solution for $a(u)$ is given by:
\be a(u) = -\left(\frac{2}{-c}\right)^{2/3} \zeta\left( \left(\frac{-c}{2}\right)^{1/3} (u-u_0);\ 0,\, -c^2 \right) + a_0,\ee
where $\zeta(z;g_2,g_3)$ is the Weierstrass zeta function. This function has an
infinite number of poles, so the solution will not be globally
well-defined. For this solution (as mentioned above), $(P_z^f)^2\neq 0$.

\subsubsection{$\la{sl}(3)$ Example 2: not nilpotent or semisimple, $(P_z^f)^2\neq0$}\label{sec:potformsl3nssnoSUSY}
Let us consider the monodromy:
\be X = \left(\begin{array}{ccc} 1 & 0 & 1\\ 0&-2&0\\ 0&0&1\end{array}\right).\ee
This monodromy is neither nilpotent nor semisimple, and moreover its semisimple and nilpotent parts can clearly never be embedded into two commuting $\la{sl}(2)$ algebras (as $\la{sl}(3)$ does not contain two commuting $\la{sl}(2)$ subalgebras).

The potential is given by:
\be W = -6 -\frac92 e^{-\phi_1}\left(e^{\sqrt{3}\phi_2} a_1^2 + e^{-\sqrt{3}\phi_2}a_3^2\right) - \frac12 e^{-2\phi_1}(1+3a_1a_3)^2 < 0,\ee
so the potential is always strictly negative. We note that $X+X^T$ is conjugate to $\diag(1,3,-4)$, which does not meet the criteria given in (\ref{eq:potformSUSYcond}) as the non-zero eigenvalues are not paired, so it will not be possible to find a solution with $(P_z^f)^2=0$.

A particularly easy ``diagonal'' solution is given by:
\begin{align}
\begin{split}
 a_1 &=a_2=a_3=0,\qquad
 \phi_2 = c,\qquad
 \phi_1 = \log\frac{\sinh(2\sqrt{3} u)}{2\sqrt{3}}.
\end{split}
\end{align}

\subsubsection{$\la{sl}(3)$ Example 3: not nilpotent or semisimple, $(P_z^f)^2=0$}\label{sec:potformsl3nssSUSY}
We can alter the previous monodromy slightly to:
\be X = \left(\begin{array}{ccc} 1 & 0 & \pm 2\\ 0&-2&0\\ 0&0&1\end{array}\right).\ee
Now, we see that the $X+X^T$ is conjugate to $\diag(4,-4,0)$, which does fit the form given in (\ref{eq:potformSUSYcond}) since the non-zero eigenvalues are paired, so it should be possible to find a solution with $(P_z^f)^2=0$.

The potential is given by:
\be W = -6 - \frac12 e^{-2\phi_1} (3a_1a_3 + x)^2 - \frac92 e^{-\phi_1}(e^{\sqrt{3}\phi_2} a_1^2 + e^{-\sqrt{3}\phi_2}a_3^2) < 0,\ee
where $x=\pm 2$ is the $(1,3)$ entry of $X$.

To find a solution with $(P_z^f)^2=0$, we can start with the block diagonal form given by (\ref{eq:potformSUSYcond}) for $\dot{m}(u_0)$ for $m(u_0)=1$ and find the solution order by order in $u-u_0$; doing this will suggest the ansatz $\phi_2=0, a_1=a_3$. Then, we can easily solve the full equation $(P_z^f)^2 =0$ to find the following solution:
\begin{align}
\begin{aligned}
 \phi_1 &= \log\left[2-\cosh(2\sqrt{3}u)\right], & \qquad \phi_2 &=0,\\
 a_1 &= \pm a_3 = 2\sqrt{\frac23} \sinh(\sqrt{3}u), & a_2 &= \pm \frac43 \sinh^2(\sqrt{3}u).
\end{aligned}
\end{align}
One can easily check that this solution satisfies the equations of motion as well. For $\phi_1$ to be real, the range of $u$ is restricted to:
\be\label{eq:potformurange} -\frac{1}{2\sqrt{3}} \log(2+\sqrt{3}) < u < \frac{1}{2\sqrt{3}} \log(2+\sqrt{3}),\ee
so in particular this solution does not extend to the origin $r=0\, (u\rightarrow+\infty)$. The metric function can be found from (\ref{eq:MansatzEOMU}) and is given by:
\be e^{U} = e^{-2u^2 + c_1u + c_0}\left[ -\cosh(2\sqrt{3}u) + 2\right].\ee
This is positive only in the range indicated above in (\ref{eq:potformurange}).

\subsubsection{$\la{sl}(3)$ Example 4: non-trivial nilpotent, $(P_z^f)^2=0$}\label{sec:potformsl3nilpotSUSY}
Let us now take:
\be X = X_{12} + X_{13}=\left(\begin{array}{ccc} 0 & 1 & 0\\ 0&0&1\\ 0&0&0\end{array}\right).\ee
Again, because $X$ is nilpotent, there exists an embedding of
$\la{sl}(2)$ into $\la{sl}(3)$ that takes $X_{\la{sl}(2)}$ onto $X$, and
the $\la{sl}(2)$ $N$-brane then trivially gives a solution with this
$X$. However, for this $N$-brane embedding, we know that $P_z^f$ is
conjugate to $X$, so that we will have $(P_z^f)^2\neq 0$. Here, we would
like to investigate the possibility of a solution with the embedding
given by $X$ but with $(P_z^f)^2=0$.

We can again find a solution by demanding $m(u_0)=1$, and then we find:
\be \dot{m}(u_0) = \frac{1}{\sqrt{2}} \left(\begin{array}{ccc} 1 & 0 & 1\\ 0&-2&0\\ 1&0&1\end{array}\right),\ee
and we can use this to perturbatively find a solution in around $u_0=0$:
\begin{align}
\begin{aligned}
  a_1 &= a_3 = 0,\\
 a_2 &= \frac{1}{\sqrt{2}}\left( u - \frac12 u^3 + \frac{3}{20} u^5 - \frac{3}{56}u^7 + \frac{1}{56}u^9 - \frac{15}{2464}u^{11}  \cdots\right),\\
 \phi_1 &= -\frac34 u^2 - \frac{3}{16} u^4 - \frac{3}{16} u^6 - \frac{117}{896}u^8 - \frac{123}{1120} u^{10} - \frac{165}{1792} u^{12} - \cdots,\\
 \phi_2 &= \sqrt{\frac32}\left( u + \frac12 u^3 + \frac{3}{10}u^5 + \frac{3}{14} u^7 + \frac{19}{112} u^9 + \frac{87}{616} u^{11} + \cdots\right).
\end{aligned}
\end{align}
We have been unable to find an analytic closed-form expression for these functions, but from this perturbative expansion it seems likely that the range of $u$ will again be restricted---this time to $\|u\|\leq 1$.

\subsection{Taxonomy of Brane Solutions}\label{sec:potformtaxonomy}
We have now explored a plethora of single-centered $\la{sl}(n)$ brane solutions with the ansatz (\ref{eq:Mansatz}) with varying properties. These brane solutions are summarized in a self-contained way in Table \ref{tab:taxonomy}. In this table, we give (i) the subalgebra of $\la{e}_8$ that the solution lives in, (ii) the monodromy $X$ of the solution, together with the fact if it is (iii) semisimple, (iv) nilpotent, and/or (v) compact (i.e.\ $X\in \la{k}$). We also give (vi) the reference to the section in the text where the explicit solution is to be found. Finally, a number of properties of the solution are given: whether (vii) $(P_z^f)^2=0$ (which is necessary but not sufficient for supersymmetry as explained in section \ref{sec:condsusy}), whether (viii) the solution is globally defined (i.e.\ extends to the origin $r=0$, at which the brane is expected to ``sit''), and whether (ix) the solution approaches a point where the potential energy vanishes, $W=0$ (see section \ref{sec:potform}).

\begin{landscape}
\begin{table}[h]
\begin{center}
 \begin{tabular}{|c|c|c|c|c||c||c|c|c|}
 \hline
 subalgebra & $X$ & semisimple & nilpotent & compact & sol. (sect.) & $(P_z^f)^2=0$ & global & $ W \rightarrow 0$\\
 \hline\hline
 $\la{sl}(2)$ & $\left(\begin{array}{cc}0&1\\0&0\end{array}\right)$ & \NO & \YES & \NO & \ref{sec:Nbrane} & \YES & \YES & \INF\\
 \hline
 $\la{sl}(2)$ & $\left(\begin{array}{cc}0&1\\-1&0\end{array}\right)$ & \YES & \NO & \YES & \ref{sec:Sbrane} & \YES & \YES & \YES\\
 \hline
 $\la{sl}(2)$ & $\left(\begin{array}{cc}1&0\\0&-1\end{array}\right)$ & \YES & \NO & \NO & \ref{sec:Hbrane} & \YES & \NO & \NO\\
 \hline\hline
 $(\la{sl}(2)\subset\,)\,\la{e}_8$ & any nilpotent & \NO & \YES & \NO & \ref{sec:nilpotentcharges} & \YES$^{\dagger}$ & \YES & \INF\\
 \hline
 $(\la{sl}(2)^n\subset)\,\la{e}_8$ & any semisimple & \YES & \NO & \MAYBE & \ref{sec:semisimplecharges} & \YES$^{\dagger}$ & \MAYBE & \MAYBE \\
 \hline
 $\la{sl}(2)^n$ & any (in $\la{sl}(2)^n$) & \MAYBE & \MAYBE & \MAYBE & \ref{sec:othercharges} & \YES$^{\dagger}$ & \YES & \MAYBE\\
 \hline\hline
 $\la{sl}(3)$ & $\left(\begin{array}{ccc}0&0&1\\0&0&0\\0&0&0\end{array}\right)$ & \NO & \YES & \NO & \ref{sec:potformsl3X13} & \NO & \NO & \NO\\
 $(\la{sl}(2)\subset\la{sl}(3))$ &  & \NO & \YES & \NO & \ref{sec:Nbrane} & \YES$^{\ddagger}$ & \YES & \INF\\
 \hline
 $\la{sl}(3)$ & $\left(\begin{array}{ccc}1&0&1\\0&-2&0\\0&0&1\end{array}\right)$ & \NO & \NO & \NO & \ref{sec:potformsl3nssnoSUSY} & \NO & \YES & \NO \\
 \hline
 $\la{sl}(3)$ & $\left(\begin{array}{ccc}1&0&\pm 2\\0&-2&0\\0&0&1\end{array}\right)$ & \NO & \NO & \NO & \ref{sec:potformsl3nssSUSY} & \YES & \NO & \NO\\
 \hline
 $\la{sl}(3)$ & $\left(\begin{array}{ccc}0&1&0\\0&0&1\\0&0&0\end{array}\right)$ & \NO & \YES & \NO & \ref{sec:potformsl3nilpotSUSY} & \YES & \NO & \NO\\ 
 $(\la{sl}(2)\subset\la{sl}(3))$ & & \NO & \YES & \NO & \ref{sec:Nbrane} & \YES $^{\wedge}$ & \YES & \INF\\
 \hline
 \end{tabular}
 \caption{\label{tab:taxonomy}
 The summary of all single-center brane solutions we have considered with ansatz (\ref{eq:Mansatz}) in sections \ref{sec:simplebranes} and \ref{sec:morebranes}. See the text of section \ref{sec:potformtaxonomy} for the meaning of the columns and for discussion.\newline
 \emph{\textbf{Legend:} \YES: yes; \MAYBE: depending on the specific example, yes or no (see the relevant section in the text); \INF: the solution approaches $W=0$ for divergent $m$;\newline
 \YES $^{\dagger}$: $(P_z^f)^2=0$ in the fundamental of the smallest $\la{sl}(k)$ subalgebra which contains all of the $\la{sl}(2)^n$ subalgebras.\newline
 \YES $^{\ddagger}$: $(P_z^f)^2=0$ in the fundamental of $\la{sl}(2)$ and in the fundamental of $\la{sl}(3)$.\newline
 \YES $^{\wedge}$: $(P_z^f)^2=0$ in the fundamental of $\la{sl}(2)$ but $(P_z^f)^2\neq 0$ in the fundamental of $\la{sl}(3)$.}}
\end{center}
\end{table}
\end{landscape}

A first goal would be to classify the possible monodromies
$X\in\la{sl}(n)\subset\la{e}_8$ that can be the monodromy of a globally
defined (extending to the origin), supersymmetric (for which $(P_z^f)^2=0$ is necessary) brane
solution. We have showed that a neccesary condition for a monodromy to
have $(P_z^f)^2=0$ is that it is of the form given in
(\ref{eq:potformSUSYcond}). In fact, unfortunately, from the
examples we have considered it would not seem that there are any other
strong statements that one can make about the existence of such global,
supersymmetric solutions. In fact, the strongest conclusion that we seem
to be able to draw is that \emph{there is no ``easy'' classification of
possible supersymmetric monodromies}---to say nothing of the
non-supersymmetric case. The difficulty of classification is in contrast
to the much easier $\la{sl}(2)$ situation, where all possible
monodromies are conjugate to either the $A$-, $K$-, or $N$-brane.

Even though we have no strong statements to make about the classification of supersymmetric monodromies (besides the one neccessary criterium given in (\ref{eq:potformSUSYcond})), we can still make two observations based on the examples we have studied above, which may point towards general principles (although we are unaware of a general proof for either of them):
\begin{itemize}
 \item It seems that it is impossible to construct a globally (extending
       to the origin) defined solution that is supersymmetric
       (which needs $(P_z^f)^2=0$) when $X$ has a non-zero non-compact semisimple part
       (for example, the $A$-brane).
 \item It seems necessary (but by no means sufficient, as our examples show) that $W=0$ is reached in order for a solution to be globally well-defined and supersymmetric.
\end{itemize}

\section{Conclusions}\label{sec:conclusion}
In this paper, we have found a full classification of the supersymmetric
solutions in 3D maximal supergravity. There are two classes of
solutions: a \emph{null} class which are 1/2-BPS $pp$-waves; and a
\emph{timelike} class. For the timelike class, the amount of
supersymmetry preserved is necessarily and sufficiently determined by
the conjugacy class or orbit of the quantity $P_z$: only specific
\emph{nilpotent orbits} preserve some supersymmetry; which nilpotent
orbits preserve supersymmetry and how much is tabulated in Table
\ref{tab:orbits}.

While this classification gives all of the necessary and sufficient conditions for local supersymmetry, it lacks an interpretation in terms of the allowed (supersymmetric) scalar monodromies as we travel around branes or point particles. To obtain more insight in this matter, we have considered a number of explicit single-center point particle solutions living in a subalgebra $\la{sl}(n)\subset\la{e}_8$. By considering both $\la{sl}(2)$ and more contrived $\la{sl}(3)$ solutions with varying properties, we were able to show that there appears to be no simple relationship between the scalar monodromy and the supersymmetry-determining quantity $P_z$. This is unfortunate, because it implies that there is no easy way to use our classification to understand which classes of brane monodromies will be able to preserve supersymmetry.

One of the subtleties that arises is that in order to find the explicit solutions, we need to solve a particular ordinary
group-valued differential equation. This differential equation has various interesting properties, including the existence
of a series of conserved quantities, but unfortunately these are insufficient to determine the solutions we are interested in.
Moreover, as we demonstrated with explicit examples, it can happen that there exists a solution of the differential equation which
exists locally but which cannot be extended all the way to the origin where the exotic object is located. These solutions could
still be of value and describe e.g.\ the geometry between boundaries or domain walls, but are not bona fide gravitational descriptions
of exotic objects. Thus, the presence of supersymmetry alone does not guarantee the existence of the exotic object.

Even if one does find a supersymmetric object, it would require further analysis whether this is truly a bound state or
whether it is a marginal bound state or superposition of other more fundamental objects. This would probably require us
to start looking at multi-centered solutions which we have not attempted in this paper.

We should also emphasize that we employed a particular ansatz (\ref{eq:Mansatz}) when we constructed our solutions. Though this
ansatz appears quite natural, we did not prove that it exhausts all possible supersymmetric solutions and we leave this to future 
work. 

Another important point is that we have considered monodromies in the full supergravity U-duality group $E_{8(8)}(\mathbb{R})$, without making an effort to determine whether these monodromies would exist in the appropriate arithmetic subgroup of $E_{8(8)}$ of the underlying string theory. Although the subject of arithmetic subgroups and their conjugacy classes is a difficult mathematical problem, generically one expects these groups to be generated by exponentiating compact semi-simple and nilpotent elements. This would imply that the $A$-brane (as well as all other branes with non-compact semi-simple monodromies) might simply not exist in the full string theory, providing a possible explanation as to why the supergravity brane solution does not appear to be globally well-behaved.\footnote{We wish to thank the anonymous referee for pointing this out to us.}

Another topic which would be interesting to address is the existence of
global, multi-centered solutions. To find these, on should consider a
Riemann surface with various punctures, each of which corresponds to an
exotic object with given monodromy, and ask whether this is a consistent
solution. For example, if the Riemann surface is a two-sphere, and the
branes are sevenbranes, we know from F-theory that the allowed solutions
correspond to elliptically fibered K3's. But whether this geometrical
picture can be extended to other duality groups, or whether there exists
a more algebraic approach to this problem are all interesting but
difficult open problems \cite{Kumar:1996zx, Liu:1997mb, Curio:1998bv,
Leung:1997tw, Lu:1998sx}.

A natural place to study exotic non-geometric branes that we have
studied is in the formalism of doubled or extended field theory (for
recent reviews, see \cite{Aldazabal:2013sca, Berman:2013eva}; for
related work, see \cite{Andriot:2014uda}), where these branes would have
a geometric interpretation.  We have not used this at all in our
analysis, but it would certainly be interesting to study general exotic
branes within this framework and elucidate their nature.
In the current paper, we focused on codimension two objects in string
theory, but objects with still lower codimension are also interesting.
In particular, codimension one objects are produced by the supertube
effect \cite{Mateos:2001qs} when two codimension one objects are put
together.  Codimension one objects can be more non-geometric than
codimension two objects (see e.g.\ \cite{Hassler:2013wsa}), and
formalisms such as doubled or extended field theory may become more
relevant and useful for their study.

It has been argued that exotic branes are relevant for the microscopic
physics of black holes because they can in principle be produced by
successive supertube transitions \cite{deBoer:2010ud, Bena:2011uw,
deBoer:2012ma}.  We hope that the analysis in this paper will be useful
as an approximate near-brane description of exotic branes thus produced
in black hole systems.

Overall, three dimensional maximal supergravity has turned out to be a
remarkably rich theory. Presumably, we have only scratched the surface
and many more interesting facts are still awaiting discovery.


\section*{Acknowledgments}

We thank Marco Baggio, Iosif Bena, Eric Bergshoeff, Mariana Gra\~na,
Yoshifumi Hyakutake, Tetsuji Kimura, Ruben Minasian, Eric Opdam, Fabio
Riccioni, Orestis Vasilakis, Nick Warner, and Satoshi Yamaguchi for
fruitful discussions.
The work of MS was supported in part by Grant-in-Aid for Young
Scientists (B) 24740159 from the Japan Society for the Promotion of
Science (JSPS)\@. This work is part of the research programme of the Foundation for Fundamental Research on Matter (FOM), which is part of the Netherlands Organisation for Scientific Research (NWO).

\appendix

\section{Construction of $\la{e}_8$}\label{sec:appconstruction}

\subsection{Gamma Matrices of $\la{so}(16)$}\label{app:appgammaso16}
Let us recall some facts about the $SO(16)$
algebra \cite{Marcus:1983hb, Green:1987sp}.  The Dirac algebra of
$SO(16)$ requires 256-dimensional matrices corresponding to the
reducible spinor representation ${\bf 128+128'}$ of $SO(16)$.  These
matrices can be taken in the block diagonal form as
\begin{align}
 \Gamma^I=\begin{pmatrix}
	   0& \Gamma^I_{A\dot{A}} \\
	   \Gamma^I_{\dot{B}B}&0 \\
	  \end{pmatrix},
\end{align}
where $\Gamma^I_{\dot{A}A}=(\Gamma^I_{A\dot{A}})^T$.  The range of the
indices is $I=1,\dots,16$; $A,B=1,\dots,128$;
$\dot{A},\dot{B}=1,\dots,128$.  The Clifford algebra
$\{\Gamma^I,\Gamma^J\}=2\delta^{IJ}$ amounts to the equations
\begin{align}
 \Gamma^I_{A\dot{A}}\Gamma^J_{\dot{A}B}
 +\Gamma^J_{A\dot{A}}\Gamma^I_{\dot{A}B}
 = 2\delta^{IJ}\delta_{AB},\qquad
 \Gamma^I_{\dot{A}A}\Gamma^J_{A\dot{B}}
 +\Gamma^J_{\dot{A}A}\Gamma^I_{A\dot{B}}
 =2\delta^{IJ}\delta_{\dot{A}\dot{B}}.
\end{align}
We also define
\begin{align}
 \Gamma^{IJ}_{AB}
 \equiv \half(\Gamma^{I}_{A\dot{A}}\Gamma^{J}_{\dot{A}B}
 -\Gamma^{J}_{A\dot{A}}\Gamma^{I}_{\dot{A}B}).
\end{align}

We will always be working in a purely real representation of the $\Gamma$ matrices, so that the Majorana spinors $\epsilon^I$ (the supersymmetry transformations) are always real.

\subsection{Two Constructions of $\la{e}_{8(8)}$}\label{app:twoconstructions}
The algebra of $\la{e}_8$ can be viewed as appending 128 generators $Y^A$ in the Majorana-Weyl representation of $\la{so}(16)$ to the 120 generators $X^{IJ}$ of $\la{so}(16)$ \cite{Green:1987sp}. All the commutation relations are then given by:
\begin{align}
\nonumber [ X^{IJ}, X^{KL}] &= \delta^{IK} X^{IL} + \delta^{JL}X^{IK} - \delta^{IL}X^{JK} - \delta^{JK} X^{IL},\\
\nonumber [ X^{IJ}, Y^A ] &=  -\frac12\Gamma^{IJ}_{AB} Y^B,\\
\label{eq:e8commutators} [Y^A, Y^B] &=  \frac12\Gamma^{IJ}_{AB} X^{IJ}.
 \end{align}
The sign of the commutators was chosen so that the generators $X, Y$ are
those of the real algebra of $\la{e}_{8(8)}$, where the compact
(anti-hermitian) generators are given by $X^{IJ}$ and the non-compact
(hermitian) generators are given by $Y^A$. Thus, the generators of
$\la{e}_{8(8)}\ominus \la{so}(16)$, used ubiquitously in this paper, are the $Y^A$'s.

There is also a second way of constructing the algebra $\la{e}_8$,
in the canonical way of the root decomposition in the Cartan-Weyl basis. The roots of
$\la{e}_8$ can be taken to be all 8-vectors of two kinds: one with
two non-zero components that are each given by either $+1$ or $-1$ (112
roots of this kind) and a kind with all 8-vectors with all components
equal to $\pm 1/2$ with the sum of the components being $2n$ with $n$ an
integer, i.e.\ there are an even number of $+1/2$ components (128 roots
of this kind). Schematically, we have:
\be \label{eq:e8roots} \Delta = \{ \alpha=e_i\pm e_j; i\neq j \} \cup \{ \alpha=\frac12\left(e_1\pm e_2\pm \cdots\pm e_8\right), \sum\alpha^i\mod 2=0 \},\ee
with $\Delta$ the collection of roots, and $e_i$ the 8-vector with an entry 1 in the $i$th place and 0 everywhere else. The algebra then consists of eight Cartan generators
$H^i$ and the 240 root generators $E_{\alpha}$. All commutators are
given by:
\begin{align}
 [H_i, E_{\alpha}] &= \alpha^i E_{\alpha},\\
[ E_{\alpha}, E_{-\alpha}] &= \alpha_i H^i,\\
[ E_{\alpha}, E_{\beta}] &= N_{\alpha,\beta} E_{\alpha+\beta},
\end{align}
where $N_{\alpha,\beta}$ is only non-zero if $\alpha+\beta$ is a root; these coefficients $N_{\alpha,\beta}$ are highly constrained by a number of symmetries that the commutators need to respect (including Jacobi identities), but there are still a number of constants in their definition that one can choose arbitrarily. With this construction, the 128 non-compact, hermitian generators of $\la{e}_{8(8)}$ are the generators $H^i$ and $E_{\alpha}+E_{-\alpha}$; the 120 compact, anti-hermitian generators are given by $E_{\alpha}-E_{-\alpha}$.

\subsection{Map Between Constructions}
The root decomposition is natural from a Lie algebra point of view; for
example the nilpotent orbit representatives in \cite{Djokovic:reps} are
given in terms of sums of the root vectors $E_{\alpha}$. However, from a
3D supergravity point of view, the decomposition in $\la{so}(16)$+spinor
is more natural, as e.g.\ the elements $P$, which live in the
non-compact part of the algebra, are always treated as a spinor of
$\la{so}(16)$. Thus, for some calculations we need an explicit map
between the two constructions of $\la{e}_{8(8)}$. We will now
schematically explain our way of mapping the $\la{so}(16)$+spinor
construction to the root decomposition; we will essentially be using the
adjoint representation of $\la{e}_8$ for this.

The first step in our construction of the map is to select 8 generators to serve as the Cartan subalgebra $\{H^I\}$. These must be non-compact generators, so we can select 8 linear combinations of generators $Y^A$ that all mutually commute, so that $H^i = \mathcal{H}^i_A Y^A$. In practice, we do this by selecting a first generator $H^1(Y^A)$ arbitrarily, and then iteratively finding a new generator $H^j(Y^A)$ that commutes with all the previous $H^i$'s. The commutator of $H^i$ with an arbitrary element $x^{IJ} X^{IJ} + P^A Y^A$ is given by the matrix operation $C^i$ on $(x^{IJ}, P^A)$ as:
\be C_i \left(\begin{array}{c} x^{ij}\\ P^A \end{array}\right)=\frac12\left( \begin{array}{cc} 0_{120} & \mathcal{H}^i_A \Gamma^{IJ}_{AB}\\ \mathcal{H}^i_A \Gamma^{IJ}_{AB} & 0_{128} \end{array}\right)\left(\begin{array}{c} x^{IJ}\\ P^B \end{array}\right),\ee
where the $0$ denote square zero matrices of the denoted dimension, and the upper-right and lower-left matrices should be considered as $128\times120$ and $120\times128$ matrices, respectively. Finally, the generators $H^i$ must all be \emph{semi-simple}, which is equivalent to demanding that the matrices $C^i$ are diagonalizable. This must be checked to safely conclude that $\textrm{span}\{H^i\}$ is a valid Cartan subalgebra of $\la{e}_8$. Finally, we rescale all of our $H^i$'s so that the 128-vector norm of $(\mathcal{H}^i_A)$ is 1.

The second step is to identify the 240 root vectors $E_{\alpha}$. This requires finding the 240 simultaneous eigenvectors of the 8 matrices $C^i$ defined above. This can again be done iteratively. We can start by diagonalizing $C^1$ as $C^1 = P_1 D_1 P_1^{-1}$, where $D_1$ is the resulting diagonal matrix. Then, we perform the change of basis given by $P_1$ on $C^2$ and diagonalize this by $P_1^{-1} C^2 P_1 = P_2 D_2 P_2^{-1}$,
and so forth. In the end, the matrix that simultaneously diagonalizes all $C^i$'s (and thus contains the sought-after 240 eigenvectors) is given by $P \equiv P_1\cdot P_2\cdot \ldots\cdot P_8$. The root vectors $(\mathcal{E}_{\tilde{\alpha}, ij},\mathcal{E}_{\tilde{\alpha}, A})$, with $E_{\tilde{\alpha}} = \mathcal{E}_{\tilde{\alpha}, A} Y^A + \mathcal{E}_{\tilde{\alpha}, ij} X^{ij}$, are then given by the columns of $P$. By solving the appropriate eigenvalue equation for each column of $P$, we can match each column to a root vector $\tilde{\alpha}$.

We are not yet finished, as the roots $\tilde{\alpha}$ that we obtain by this procedure are in general not the ones we are interested in as described above. This is because there are many different admissable root systems for $\la{e}_8$, which are related to a rotation of the Cartan generators $H^i$. So the final step is to find an $8\times 8$ orthogonal matrix $T$ which rotates the $H^i$'s as $H\rightarrow T H$, and thus rotating the $\tilde{\alpha}$'s into the correct root vectors $\alpha=T\tilde{\alpha}$.

Finally, we rescale all vectors $(\mathcal{E}_{\tilde{\alpha}, ij},\mathcal{E}_{\tilde{\alpha}, A})$ so that the 248-vector norm is $1$ for each of these vectors, and further we rescale all negative root vectors (i.e.\ we must make a choice of positive and negative roots for every pair of roots $\pm\alpha$) with $\pm1$ so that $[E_{\alpha},E_{-\alpha}]=+\alpha\cdot H$ for every $\alpha$. With this normalization, we also have $[E_{\alpha}, E_{\beta}] = N_{\alpha,\beta} E_{\alpha+\beta}$ with $N=\pm1$ if $\alpha+\beta$ is a root.

Note that, by construction (i.e.\ by using the commutation relations (\ref{eq:e8commutators}) and selecting the $H^i$'s as combinations of $Y$'s), we have $\mathcal{E}_{-\alpha, IJ}=-\mathcal{E}_{\alpha, IJ}$ and $\mathcal{E}_{-\alpha,A}=+\mathcal{E}_{\alpha, A}$, so that indeed $E_{\alpha}+E_{-\alpha}$ and $H^i$ are the non-compact generators and $E_{\alpha}-E_{-\alpha}$ are the compact generators.

\subsection{Calculations in $\la{e}_{8(8)}$}
There are a number of calculations in the paper where we need the explicit map between the Cartan decomposition and the $\la{so}(16)$+spinor construction of $\la{e}_{8(8)}$. Here we will briefly sketch the procedure of these calculations.

\subsubsection{Calculating Supersymmetry Preserved by $P_z$ (proof of result \ref{prop:Porbits} in section \ref{sec:proofPorbits})}
\label{sec:appproofPorbits}
Our goal is to calculate directly the amount of supersymmetry that the equation
\be\label{eq:appPeq} \overline{\zeta}^I\Gamma^I_{A\dot{A}}P^A_z = 0\ee
preserves for $P_z$ in a particular nilpotent $K(\bbC)$-orbit of $\la{p}_\bbC$. A list of explicit representatives for each $G(\bbR)$-orbit in $\la{g}_\bbR$ is listed in \cite{Djokovic:reps}; moreover, these representatives are the nilpositive elements of Cayley triples (see section \ref{sec:lieconceptstriples}). This means that, given a nilpositive element $X$ which is a representative of the orbit $\mathbf{i}$ (where $\mathbf{i}\in\{\mathbf{1},\ldots,\mathbf{115}\}$), we can find the rest of the standard Cayley triple by $Y=-\theta(X)$ and $H=[X,Y]$. Then, the fact that $\{H,X,Y\}$ is a standard Cayley triple and the Kostant-Sekiguchi bijection (see section \ref{sec:lieconceptsorbits}) assures us that $X'=\frac12(X+Y+iH)$ is a nilpotent element in the corresponding $K(\bbC)$-orbit $\mathbf{i}$ in $\la{p}_\bbC$.

To summarize, given the element $X\in\la{g}_\bbR$ listed in \cite{Djokovic:reps}, we take $P_z = \frac12(X - \theta(X) - i[X,\theta(X)])$. For this, we need the map between the root decomposition and the $\la{so}(16)$+spinor construction as $X$ is listed in \cite{Djokovic:reps} in terms of the root vectors $E_{\alpha}$. Once we have constructed a representative $P_z$ in this way, it is then a simple manner of plugging in this $P_z$ in (\ref{eq:appPeq}) and solving for the 16 complex components $\zeta^I$ to see how many (real) degrees of freedom of $\zeta^I$ are still free to choose and thus how much supersymmetry is preserved.

\subsubsection{Proving $c_i H_i$ Preserves No Supersymmetry (proof of result \ref{prop:Pnilpotent} in section \ref{sec:proofPnilpotent})}
\label{sec:appproofPnilpotent} 
We wish to prove that the equation (\ref{eq:appPeq}), for $P_z = \sum c_i H_i$, always breaks all supersymmetry unless all of the complex constants $c_i=0$. Again, we need the map between the root decomposition and the $\la{so}(16)$+spinor construction for this calculation. Then it is a simple matter to set $P_z = \sum c_i H_i$ and solve (\ref{eq:appPeq}) in Mathematica; the solution unambiguously tells us that either all components $\zeta^I = 0$ (supersymmetry is completely broken) or all constants $c_i = 0$ ($P_z=0$).

\subsubsection{Finding Brane Representatives for the Orbits (section \ref{sec:branereps})}\label{sec:appbranerep}
Given a brane configuration $X$, we need an unambiguous way to determine which orbit the brane configuration is a representative of. Fortunately, there are two easy criteria, which are two orbit-invariant quantities given by $\dim$ and $\inv$. The dimension $\dim$ of an orbit is the dimension as an algebraic variety, and can be calculated using a representative $X$ as \cite{collingwood}:
\be \dim( \mathcal{O}_X) = \dim(\la{g}) - \dim(\la{g}^X),  \ee
where the dimensions on the right hand side are simply the dimensions as vector spaces, and $\la{g}^X = \{Y\in\la{g}, [X,Y]=0\}$ is the centralizer of $X$ in $\la{g}$ (which is a subalgebra of $\la{g}$). The invariant $\inv$ (of a real nilpotent orbit) is defined as the (complex) dimension of the centralizer, restricted to the compact space $\la{so}(16)_\bbC$, of the neutral element $H'=i(X-Y)$ of the Cayley transform (see section \ref{sec:lieconceptstriples}) of a real Cayley triple $\{H,X,Y\}$ \cite{Djokovic:reps}. In symbols, for any representative $X$ that is the nilpositive element of a Cayley triple:
\be \inv(\mathcal{O}_X) = \dim\left( \la{so}(16)^{H'}\right), \qquad H' = i(X-Y) = i(X+\theta(X)).\ee
Since our brane configurations are always sums of (positive) root vectors $E_{\alpha}$, they are by construction the nilpositive element of a Cayley triple, so that $\inv$ is easily calculated for these brane configurations. We note that some higher dimensional nilpotent orbits are not unambiguously specified by giving $\dim$ and $\inv$ of the orbit, but these two characteristics are sufficient for the orbits that we are interested in ($<\mathbf{20}$).


\section{Mathematical Theorems}
\label{sec:theorems}
Here we review theorems and facts from various mathematical papers for convenient reference. We will refer to them all as lemmas in our paper, even though they are lemmas, propositions, or theorems in the original paper.\\

The paper \cite{KostantRallis71} deals mainly with the subject of $K(\bbC)$-orbits in $\la{p}_\bbC$, and is thus invaluable to our discussion of $P_z$. We will leave out the subscript $\mathbb{C}$ in the statement of the following propositions. Also, the following three propositions are valid for any Lie algebra $\la{g}$ with Cartan decomposition $\la{g}=\la{k}\oplus \la{p}$.

\begin{mathproposition}\label{lemma:KRxinp}
(\cite{KostantRallis71}, Prop. 3 (p764).) If $x\in \la{p}_\bbC$ has a Jordan decomposition $x=x_s+x_n$, then $x_s\in \la{p}_\bbC$ and $x_n\in \la{p}_\bbC$.
\end{mathproposition}
Here and in the following, the Jordan decomposition of an element $x$ is the unique decomposition $x=x_s+x_n$ such that $x_s$ is semi-simple, $x_n$ is nilpotent, and $[x_s,x_n]=0$. The following two lemmas are quite analogous to their counterparts for $G(\bbC)$-orbits in $\la{g}_\bbC$.

\begin{mathproposition}\label{lemma:KRxsinclosure}
(\cite{KostantRallis71}, Lemma 11 (p782).) Let $x\in\la{p}_\bbC$ have a Jordan decomposition $x=x_s+x_n$. Then $x_s\in \overline{\mathcal{O}_x}$, where $\mathcal{O}_x$ is the orbit of $x$.
\end{mathproposition}

\begin{mathproposition}\label{lemma:KRxsconjugatecartan}
(\cite{KostantRallis71}, Theorem 1, partial (p764).) All Cartan subspaces of $\la{p}_\bbC$ are conjugate under the action of $K(\bbC)$. Moreover, an element in $x\in\la{p}_\bbC$ is semi-simple if and only if it can be embedded in a Cartan subspace of $\la{p}_\bbC$. In particular, $x$ is semi-simple if and only if $x$ is $K$-conjugate to an element in a given Cartan subspace.
\end{mathproposition}

In the paper \cite{Djokovic:hasse03} (and erratum \cite{Djokovic:hasseerr}), the complete Hasse diagram of the nilpotent $K(\bbC)$-orbits in $\la{p}_\bbC$, or equivalently the $G(\bbR)$-orbits in $\la{g}_\bbR$, is found for the split form $G=E_{8(8)}$. The Hasse diagram indicates diagramatically how the nilpotent orbits' closures are contained in one another. Here we reproduce the first part of this closure diagram diagramatically as a lemma. The closure diagram should be read as follows: if there is a line going from orbit $\mathcal{O}$ to an orbit $\mathcal{O}'$ to the right of $\mathcal{O}$, then this means that $\overline{\mathcal{O}}\subset \overline{\mathcal{O}'}$.
\begin{mathproposition}\label{lemma:E8hasse} (\cite{Djokovic:hasse03,Djokovic:hasseerr}, Theorem 4.2 (p580).)
The first part of the Hasse diagram for the $K(\bbR)$-orbits in $\la{p}_\bbR$ of $E_{8(8)}$ is given by Fig.\ \ref{fig:hassediagram}. In addition, all other nilpotent orbits not depicted are connected (directly or indirectly) to orbit \textbf{5}, so that all nilpotent orbits not depicted (as well as other orbits which are depicted, such as orbit \textbf{8}) contain orbit \textbf{5} in their closure.
\end{mathproposition}
\begin{figure}[h]
\begin{center}
\includegraphics[width=\textwidth]{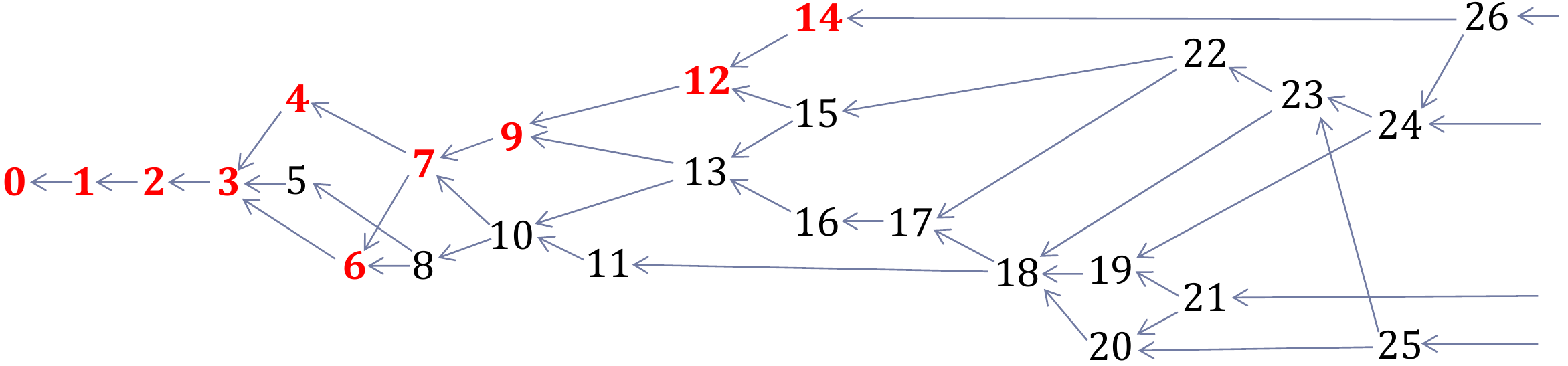}
\caption{\label{fig:hassediagram}The first part of the Hasse closure
diagram for nilpotent orbits in $E_{8(8)}$. Boldface (red) orbits are
those that preserve some supersymmetry (see Table \ref{tab:orbits};
$\mathbf{0}=\{0\}$ is the trivial orbit).}
\end{center}
\end{figure}

In discussing semi-simple monodromies, we need results on the
$G(\bbR)$-orbits of semi-simple elements in $\la{g}_\bbR$.

For the complexified version, these objects are very simple as every
Cartan subalgebra is $G(\bbC)$-conjugate in $g_\bbC$, so every
semi-simple element is conjugate to an element in a given Cartan
subalgebra; the complex semi-simple orbits are in one-to-one
correspondence with the elements of a Cartan subalgebra.

However, the situation is more complicated for real algebras. In
\cite{sugiura59}, all of the conjugacy classes of Cartan subalgebras in
real Lie algebras are classified and explicitly found:
\begin{mathproposition}\label{lemma:sugiuracartanconjugacy} (\cite{sugiura59}, mainly ``Type (E VIII)'' in paragraph 4 (p424-426).)
All Cartan sub-algebras in $\la{e}_{8(8)}$ are conjugate to one of
\emph{ten} Cartan subalgebras which are given by the following
construction.

Consider the roots $\alpha_1=(1,0,0,0,0,0,1,0)$,
$\alpha_2=(1,0,0,0,0,0,-1,0)$, $\alpha_3 = (0,1,0,0,0,0,0,1)$,
$\alpha_4=(0,1,0,0,0,0,0,-1)$, $\alpha_5 = (0,0,1,1,0,0,0,0)$,
$\alpha_6=(0,0,1,-1,0,0,0,0)$, $\alpha_7=(0,0,0,0,1,1,0,0)$,
$\alpha_8=(0,0,0,0,1,-1,0,0)$. Note that these all satisfy that pairwise
$\alpha_i\pm \alpha_j$ is not a root for $i\neq j$. For each of these 8
roots $\alpha_i$, we consider the elements $S_{\alpha_i} = E_{\alpha_i}
- E_{-\alpha_i}$ and $H_{\alpha_i} = \alpha_i \cdot \vec{H}$ (where
$\vec{H}$ are our Cartan generators constructed in $\la{e}_{8(8)}\ominus
\la{so}(16)$). Every Cartan subalgebra is conjugate to a Cartan
subalgebra spanned by some combination of $H_{\alpha_i}$'s and
$S_{\alpha_i}$'s.  (For the explicit form of the ten choices, see
\cite{sugiura59}.)

In particular, because every semi-simple element is part of a Cartan sub-algebra, this implies that every semi-simple element in $\la{e}_{8(8)}$ is conjugate to some linear combination $\sum_i c_i H_{\alpha_i} + \sum_j  c_j S_{\alpha_j}$ where $i$ and $j$ do not take the same values.
\end{mathproposition}

\bibliographystyle{JHEP}
\bibliography{Papers}

\end{document}